\documentclass[letterpaper,twocolumn,english,preprintnumbers,amsmath,amssymb,superscriptaddress,footinbib,prx]{revtex4-2}
\usepackage{graphicx} 

\usepackage{catchfile}

\usepackage{amsthm}
\usepackage{amssymb}

\graphicspath{{./Images/}}
\usepackage{xcolor}
\usepackage[colorlinks = true, citecolor=blue, linkcolor=black,urlcolor=purple]{hyperref}
\usepackage{enumitem}
    \setlist[itemize]{noitemsep, topsep=0pt}
    \setlist[enumerate]{noitemsep, topsep=0pt}
\usepackage{verbatim}
\usepackage{mathtools,amssymb,amsmath}
\usepackage{bm}
\usepackage{yhmath}
\usepackage[inline]{asymptote}
\usepackage{cancel}
\usepackage{relsize}
\usepackage{array}
\usepackage{bm}
\usepackage{subfigure}

\newcommand{\ket}[1]{|#1\rangle}

\newcommand{\be}{\begin{equation}}
\newcommand{\ee}{\end{equation}}
\newcommand{\<}{\langle}
\renewcommand{\>}{\rangle}

\newcommand{\Tr}{{\rm Tr\,}}

\usepackage[normalem]{ulem}
\usepackage{color}

\definecolor{crimson}{RGB}{220, 20, 60}

\usepackage{afterpage}

\newcommand{\ar}{a}
\newcommand{\br}{b}
\newcommand{\crf}{c}
\newcommand{\dr}{d}
\newcommand{\er}{e}

\newcommand{\arc}{a.}
\newcommand{\brc}{b.}
\newcommand{\crc}{c.}
\newcommand{\drc}{d.}
\newcommand{\erc}{e.}
\newcommand{\frc}{f.}

\newcommand{\ibmalmaden}{IBM Quantum, IBM Research -- Almaden, San Jose CA, 95120, USA}
\newcommand{\ibmyorktown}{IBM Quantum, IBM T.J. Watson Research Center, Yorktown Heights, 10598, USA}

\newcommand{\deeusc}{Department of Electrical Engineering, Viterbi School of Engineering,
University of Southern California, Los Angeles, CA 90089, USA}
\newcommand{\ibmcambridge}{IBM Quantum, MIT-IBM Watson AI Lab,  Cambridge MA, 02142, USA}
\newcommand{\yale}{Department of Physics, Yale University, New Haven CT, 06520, USA}

\usepackage{algorithm}
\usepackage{algpseudocode}

\renewcommand{\vec}[1]{{\bf #1}}
\newtheorem{prop}{Proposition}

\makeatletter
\newsavebox{\@brx}
\newcommand{\llangle}[1][]{\savebox{\@brx}{\(\m@th{#1\langle}\)}
  \mathopen{\copy\@brx\kern-0.5\wd\@brx\usebox{\@brx}}}
\newcommand{\rrangle}[1][]{\savebox{\@brx}{\(\m@th{#1\rangle}\)}
  \mathclose{\copy\@brx\kern-0.5\wd\@brx\usebox{\@brx}}}
\makeatother

\renewcommand{\>}{\rangle}
\renewcommand{\vec}{\boldsymbol}
\usepackage{xcolor}

\usepackage{tikz}
\usetikzlibrary{quantikz} 
\usepackage{adjustbox} 
\usepackage{bibunits}

\newcommand{\ourtitle}{Uncovering Local Integrability in Quantum Many-Body Dynamics}

\makeatletter 
\renewcommand{\fnum@figure}{\textbf{Fig.~\thefigure}}
\def\@caption@fignum@sep{\textbf{.} }
\makeatother 

\begin{document}
\makeatother

\title{\ourtitle}

\author{Oles Shtanko$^{*\dag}$}
\affiliation{\ibmalmaden}
\altaffiliation{These authors contributed equally: Oles Shtanko, Derek Wang;  $^\dag$Corresponding authors: oles.shtanko@ibm.com (Oles Shtanko),\\ zlatko.minev@ibm.com (Zlatko Minev)}
\author{Derek S. Wang$^{*}$}
\affiliation{\ibmyorktown}
\author{Haimeng Zhang}
\affiliation{\ibmyorktown}
\affiliation{\deeusc}
\author{Nikhil~Harle}
\affiliation{\ibmyorktown}
\affiliation{Department of Physics, Yale University, New Haven CT, 06520, USA}
\author{Alireza Seif}
\affiliation{\ibmyorktown}
\author{Ramis Movassagh}
\affiliation{\ibmcambridge}
\author{Zlatko Minev$^\dag$}
\affiliation{\ibmyorktown}

\maketitle

\textbf{Interacting many-body quantum systems and their dynamics, while fundamental to modern science and technology, are formidable to simulate and understand.
However, by discovering their symmetries, conservation laws, and integrability one can unravel their intricacies. Here, using up to 124 qubits of a fully programmable quantum computer, we uncover local conservation laws and integrability in one- and two-dimensional periodically-driven spin lattices in a regime previously inaccessible to such detailed analysis. We focus on the paradigmatic example of disorder-induced ergodicity breaking, where we first benchmark the system crossover into a localized regime through anomalies in the one-particle-density-matrix spectrum and other hallmark signatures. We then demonstrate that this regime stems from hidden local integrals of motion by faithfully reconstructing their quantum operators, thus providing a more detailed portrait of the system's integrable dynamics. Our results demonstrate a versatile strategy for extracting the hidden dynamical structure from noisy experiments on large-scale quantum computers.}\\

A quantum system is entirely characterized by knowledge of its energy levels and wavefunctions---in principle. In practice, however, systems comprising numerous interacting particles lie beyond the jurisdiction of such a description. Their exponentially numerous levels, vanishing inter-level spacings, and substantial computational and resource costs \cite{Kim2023utility} render this approach impractical. Nonetheless, these complex systems often harbor a latent structure that, if discovered, could unlock a tractable description. 
In classical physics, for example, the identification of a comprehensive set of conserved quantities, or integrals of motion that are often associated with symmetries, yields such an analogous description \cite{arnol2013mathematical}.

Quantum physics presents a fundamentally more intricate scenario. Integrals of motion can exist---they are now operators that remain unchanged under system evolution and can be used to define good quantum numbers.  
When the number of these operators scales extensively with the system size, the system can be regarded as quantum integrable \cite{caux2011remarks}, which can help to unravel its complexity \cite{sutherland2004beautiful}. Unfortunately, for many interacting particles, 
even proving the existence of non-trivial integrals of motion is an intractable mathematical and numerical problem \cite{shiraishi2021undecidability}, let alone quantitatively finding them. 
Can large-scale quantum experiments help us with this problem?

Here, we answer this question in the affirmative by developing a utility-scale experimental protocol to identify conserved operators, incorporating error mitigation and relying exclusively on easily accessible local measurements. As a first application, we focus on a special class of disordered systems that exhibit many-body localization (MBL) \cite{abanin2019colloquium_mbl,basko2006metal,pal2010mbl,nandkishore2015many}. 
In these systems, pronounced disorder is believed to impede the transfer of energy and information \cite{suntajs2020quantum}, a phenomenon attributed to the emergence of a complete set of local conserved operators, or local integrals of motion (LIOMs) \cite{serbyn2013lioms,huse2014lioms,ros2015integrals,imbrie2017local}. These LIOMs offer a detailed system description, providing an alternative to the conventional eigenenergies and eigenstates framework. 
Gaining access to a complete LIOM description enables one to develop a tensor network representation of the system's eigenstates \cite{chandran2015spectral}, establishes bounds on quantum information propagation in disordered systems \cite{kim2014local, friesdorf2015local,burrell2007bounds}, and forms a crucial foundational step in creating efficient sampling algorithms \cite{ehrenberg2022simulation}.

Moreover, the study of LIOMs in experiments could add to the characterization of the finite-size MBL regime \cite{morningstar2022avalanches}, which has been of interest for numerous recent experiments \cite{rispoli2019quantum, 
lukin2019probing, roushan2017spectroscopic, choi2016exploring, schreiber2015experiment, smith2016experiment,bordia2017experiment, bordia2017experiment_2d,guo2021experiment,mi2022time,chiaro2022direct,morong2021observation}, since it  exhibits the fundamental MBL characteristics, such as the slow rate of entanglement growth and the absence of equilibration, at limited system sizes. We note that the ultimate stability of MBL in thermodynamically large systems remains a subject of ongoing debate \cite{wdroeck2017stability,morningstar2022avalanches,sels2022avalanches,leonard2023probing}. However, even for large systems where MBL is not stable, the evolution for initial times is described by a set of approximate prethermal LIOMs.\\

\begin{figure}[t!]
\setlength{\tabcolsep}{10pt}
\centering
\includegraphics[width=0.45\textwidth]
{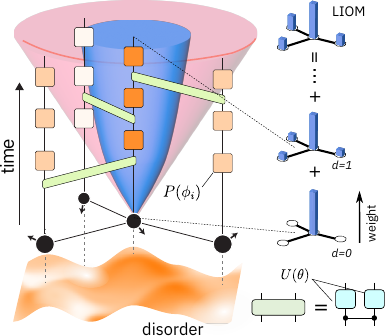}
\caption{\textbf{Many-body dynamics: quantum simulation and integrals of motion.} 
Four interacting spins (black circles) undergoing time-periodic dynamics in a modified kicked Ising model with random on-site disorder---modeled as a digital quantum circuit. Disorder is represented by single-qubit phase gates~$P(\phi_i)$ with site-randomized angles~$\phi_i$ (colored squares). Uniform spin-spin interactions (green rectangles) are mediated by controlled-Z gates followed by rotations~$U$ of strength~$\theta$, defined in Eq.~\eqref{eq:gates}. In the ergodic regime, quantum operators spread rapidly (wide cone). 
Sufficient disorder can remarkably slow down spreading the operators, leading to the many-body localization regime (blue cone). 
Disruptions to ergodicity due to integrability result in hidden conserved quantities linked to local integral of motion (LIOM) operators. We experimentally extract a complete set of LIOMs from the system's time evolution by averaging operators spread.}
\label{fig:schematics}
\end{figure}

In our experiment, we aim to reproduce MBL dynamics and extract LIOMs for 1D and 2D spin systems using circuits with up to 124 qubits and a depth of 60. Our primary goal is to test the extraction algorithms that can be used in the future to study the physics of (quasi)integrable systems. Until recently, measuring of localization  structure was possible only for small systems with low density of excitations \cite{roushan2017spectroscopic,chiaro2022direct}. Crucially, we exploit the quantum processor's large size to overcome finite-size effects in 2D, its programmability for measurements in non-commuting bases, and its operational speed to execute vast numbers of unique digital circuits (over 350,000). We address device noise through the resilience of our protocol and a tailored error-mitigation workflow \cite{Kim2023,Ci2023E,Li2017ZNE, Temme2017PECandZNE,Nation2023Mapomatic,Nation2021M3,Wallman2017Twirling,majumdar2023best,Viola1999DD, Jurcevic2021,Bravyi2021ReadoutMatrixInversion} (see Supplementary Information II.C).

Our spin system, akin to the paradigmatic kicked Ising model \cite{prozen2002general}, undergoes periodic time evolution (see Fig.~\ref{fig:schematics}), or Floquet dynamics. Unlike previous experimental work studying the MBL regime, this system does not exhibit particle conservation. Each evolution cycle is a unitary $U_F$, implemented as a series of gates on two IBM quantum processors, \texttt{ibmq\_kolkata}, \texttt{ibmq\_kyiv}, and \texttt{ibm\_washington}. 
The cycle is built from two-qubit gates (green rectangles in Fig.~\ref{fig:schematics}) comprising controlled-Z gates followed by the parameterized single-qubit gates (orange squares) 
\be\label{eq:gates}
U(\theta) \coloneqq 
\begin{pmatrix}
 \cos \theta/2  &  \sin \theta/2 \\
 \sin \theta/2  & -\cos \theta/2\\
\end{pmatrix}\;\;\mathrm{and}\;\; P(\phi_k) \coloneqq
\begin{pmatrix}
 1  &  0 \\
 0  & e^{i\phi_k}\\
\end{pmatrix}\;.
\ee
The unitary~$U$ introduces transverse field kicks, for non-zero angles~$\theta\in[0,\pi/2]$ and $P$ introduces disorder through random phases $\phi_k\in[-\pi,\pi]$ uniformly sampled over each spin site~$k$. 

For~$\theta=0$, the system dynamics are non-universal and integrable. In 1D, we expect  integrability to be present at small angles $\theta<\theta_c$ for some non-zero critical $\theta_c$ according to the MBL hypothesis. 
Disorder inhibits spread of correlations (narrow cone in Fig.~\ref{fig:schematics}), and leads to an MBL regime \cite{zhang2016floquet,ponte2015floquetmbl,bordia2017experiment}. 
In contrast, for $\theta_c<\theta\leq \pi/2$, the system's integrability is replaced by chaotic behavior leading to local operators irreversibly spreading into non-local and high-weight ones (wide cone) \cite{nahum2018operator}.

To confirm these predictions, we first subject our model to the necessary benchmarks.  
Starting with the more tractable 1D case, using small-system classical numerics, we validate the existence of a critical angle~$\theta_c\approx 0.16\pi$ that separates the ergodic and MBL regimes (see Supplementary Information I.B). In the following, we first present the corresponding experiments to demonstrate the ability of the device to reproduce MBL behavior. Next, we explore LIOMs in 1D systems and study their structure. Finally, we repeat the experiment in two dimensions  \cite{wahl2019signatures}. \\

\begin{figure*}[t]
\setlength{\tabcolsep}{10pt}
\centering
\includegraphics[width=1.0\textwidth]{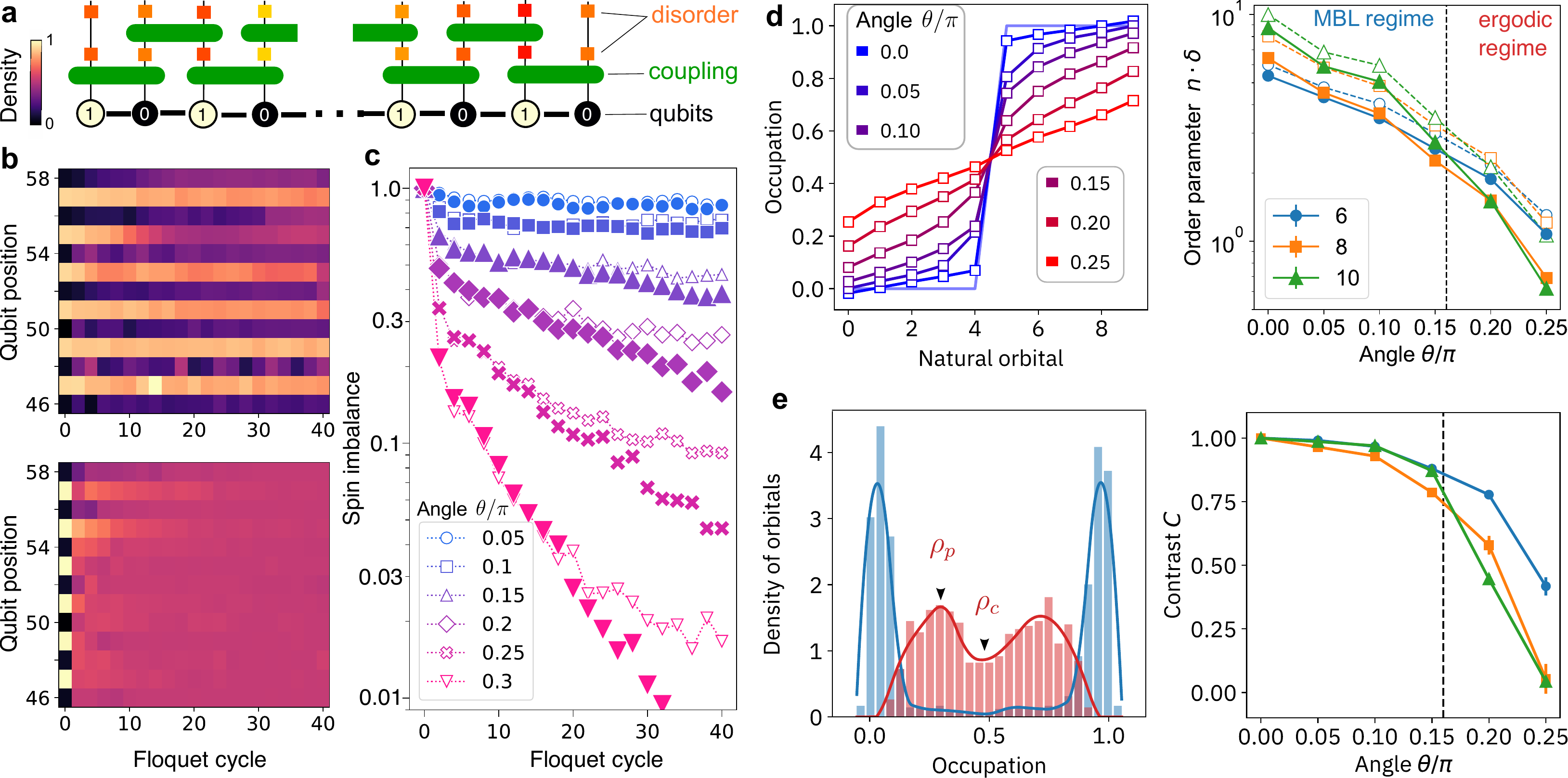}
\caption{\textbf{Conventional benchmark experiments of localization in the 1D model.}
\textbf{\arc} A 104-qubit chain initialized with an antiferromagnetic pattern evolves under a circuit comprising $\theta$-parametrized two-qubit gates (rectangles) and single-qubit disorder gates (squares). Depicted is a single Floquet cycle, which is repeated multiple times. The color of the gates indicates different disorder angles $\phi_i$.\textbf{\brc} Spin density evolution in two distinct regimes: many-body localized (MBL) at~$\theta = 0.1\,\pi$ (top) and ergodic at~$\theta = 0.3\,\pi$ (bottom).
\textbf{\crc} The spin imbalance, $\mathcal{I}$, serves as a key indicator of memory preservation from the initial state. Filled markers represent experimental data, while empty markers and dashed lines correspond to classical simulations. In the deep ergodic regime (e.g. for $\theta = 0.3\pi$), the imbalance decays rapidly within a few Floquet cycles. Conversely, in the many-body localized (MBL) regime (e.g., for $\theta = 0.1\pi$), the imbalance persists for the observation time.
\textbf{\drc} One-particle density matrix (OPDM) spectra after 9 Floquet cycles. Left: OPDM spectra for a 10-qubit chain averaged over 100 disorders; initial-state spectrum in faded blue. Residual spectrum discontinuity is conspicuous for small $\theta$. 
Right: Rescaled discontinuity across different system sizes show different scaling on different sides of classically evaluated $\theta_c\approx 0.16\,\pi$. White dots with dashed lines show the corresponding (noiseless) classical simulation. \textbf{\erc} Peak-to-gap density contrast in OPDM spectrum. Left: Observed OPDM eigenstate densities for 10 qubits in the MBL ($\theta = 0.05\,\pi$, blue) and ergodic ($\theta = 0.2\,\pi$, red) regimes. Right: Contrast as a function of $\theta$. As expected, it saturates in  the MBL regime and decreases with system size in the ergodic regime. Statistical error, calculated using bootstrapping, is represented by error bars. If not displayed, the statistical error is smaller than the marker size.}
\label{fig2:1dlocalization}
\end{figure*}

\textbf{First signature of ergodicity breaking: Memory.} We initialize a 104-qubit, 1D spin chain in a space-periodic antiferromagnetic pattern \cite{schreiber2015experiment,smith2016experiment} depicted in Fig.~\ref{fig2:1dlocalization}\ar. 
We evolve the system under a fixed disorder realization for several values of the transversal angle $\theta$ and obtain the corresponding values of individual spin polarization (see Fig.~\ref{fig2:1dlocalization}\br). The memory of the initial antiferromagnetic state is quantified by the spin imbalance, $\mathcal{I} \coloneqq \frac{n_1 - n_0}{n_1 + n_0}$, defined as the normalized difference between the average polarization of the sites initialized in state zero (spin up) and state one (spin down), denoted as $n_0$ and $n_1 \in [0,1]$, respectively. The order parameter behavior changes as the system transitions from the MBL into the ergodic regime. At small values of the angle, such as $\theta=0.05\,\pi$, the preservation of~$\mathcal{I}$ over 40 Floquet cycles demonstrates the effect of disorder that halts the quantum dynamics and preserves local memory, as expected in the MBL regime.  In contrast, at larger values, such as $\theta=0.3\,\pi$ (see Fig.~\ref{fig2:1dlocalization}\crf), the order parameter~$\mathcal{I}$ decays rapidly. This suggests that the transverse field overpowers disorder and points to an ergodic regime and  quantum information scrambling \cite{long2023phenomenology}. Unlike the systems with particle conservation \cite{lezama2019apparent}, the spin imbalance exhibits exponential asymptotics.\\ 

\textbf{Second signature: Anomalies in the one-particle density matrix (OPDM).} To obtain a more direct indicator of the MBL and ergodic regimes, we consider the spin system as a lattice of interacting hardcore bosons (i.e., site can have at most one excitation), with~$\ket{0}$ and~$\ket{1}$ representing unoccupied and occupied sites, respectively. The  indicator is provided through the OPDM \cite{bera2015opdm,bera2017opdm}, which we generalize as
\be\label{eq:OPDM}
\rho_{ij} = \<a^\dag_i a_j\> -\<a^\dag_i\>\<a_j\>\;,
\ee
where $a_i = \frac 12(X_i-iY_i)$ is the Fock annihilation operator for a hardcore boson at site $i$ and $\<\,\cdot\,\>$ represents expectation with respect to the system state. The last term is used to account for the uncorrelated dynamics in a particle-non-conserving system. This analysis of OPDM differs from previous works \cite{roushan2017spectroscopic,chiaro2022direct} as it uses discontinuity of the spectrum as an indication of crossover instead of extracting the Hamiltonian spectrum. This method provides a new tool in studying ergodicity violation in many-body systems.

Experimentally, we first prepare the system in an  antiferromagnetic many-body state, whose OPDM spectrum contains only eigenvalues one and zero, interpretable as the system  comprising natural orbitals that are either `occupied' or `empty'. We time-evolve the system and construct the OPDM from a logarithmic number of measurement combinations, where the measurement basis is randomly chosen as $X$ or $Y$ for each qubit \cite{elben2023randomized}. Owing to the limited connectivity in 1D, it is sufficient to restrict our attention to smaller qubit subchains to capture the significant correlations. 

In the left panel of Fig.~\ref{fig2:1dlocalization}\dr, we report the OPDM spectrum of 10 spins after 9 Floquet cycles, averaged over 100 disorder realizations and as a function of~$\theta$. For small~$\theta$, half the OPDM eigenstates, or natural orbitals, are nearly occupied and the other half are nearly empty---disorder has frozen the system's evolution. A conspicuous discontinuity (gap)~$\delta$ in the center of the spectrum persists despite the evolution under interactions---a robust signature of stable quasiparticles \cite{mahan2000many}. On the other hand, for large~$\theta$,  the discontinuity vanishes and the spectra look ergodic.

\begin{figure*}[t]
\setlength{\tabcolsep}{10pt}
\centering
\includegraphics[width=1\textwidth]{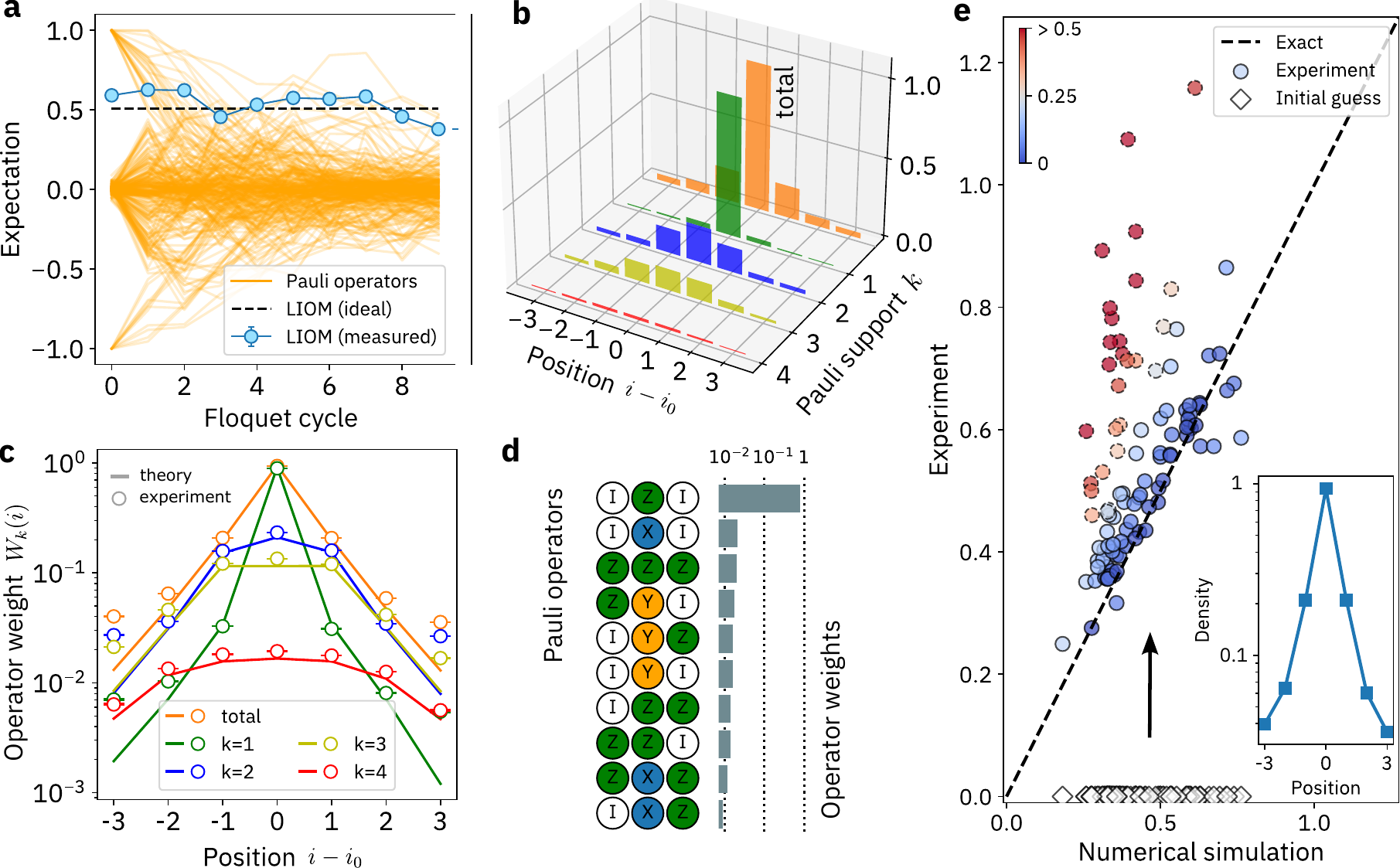}
\caption{\textbf{Uncovering local integral of motion (LIOM) operators in 1D.}
\textbf{\arc} 
Example measurements for a typical LIOM; for the initial state $\ket{000011}^{\otimes 17}\ket{00}$ and initial guess $L^{(0)} = Z_{50}$. The time-dependent expectation values of the numerous local Pauli operators composing the LIOM are displayed in orange, while the reconstructed expectations of the LIOMs are indicated by blue dots. A comparable ideal, noise-free LIOM, derived from exact diagonalization for a subsystem encompassing qubits 44 to 56, is represented by the dashed line. \textbf{\brc}
Experimental reconstruction of the average of LIOM operators in the bulk of a 104-qubit chain, extracted over 9 Floquet cycles, depicted over both space and Pauli support. Orange bars depict the total LIOM weights. Remaining bars show the individual contributions of Pauli operators with support one (green), two (blue), three (yellow), and four (red). 
\textbf{\crc} Same data (dots) compared against ideal, noiseless numerical simulation (log plot). The error bars represent bootstrapped error estimates.
\textbf{\drc} Average LIOM structure in terms of dominant Pauli operators depicted by average operator weights $w_\mu$ (log scale) on the right.
\textbf{\erc} The experimentally detected LIOM localization length compared with the same values obtained from noiseless numerical simulation. Each circle corresponds to a single LIOM and the position of the circle indicates
the experimental and simulation value of localization length. The color of experimental points indicates the relative error (see colorbar).
Dashed line is of slope one, representing ideal agreement between experiment and noise-free prediction. The alignment of the experimental points with this line, excluding outliers (circles with dashed borders), has a Pearson coefficient of 0.895. Diamonds indicate the case when LIOM is approximated by its initial guess (i.e., Pauli Z operators); localization lengths of the initial guesses (i.e. Pauli Z operators) are zero. The arrow highlights the methodological improvements achieved over this initial guess. Inset: average shape of LIOMs. }
\label{fig:1dLIOMs}
\end{figure*}

\begin{figure*}[t]
\setlength{\tabcolsep}{10pt}
\centering
\includegraphics[width=0.9\textwidth]{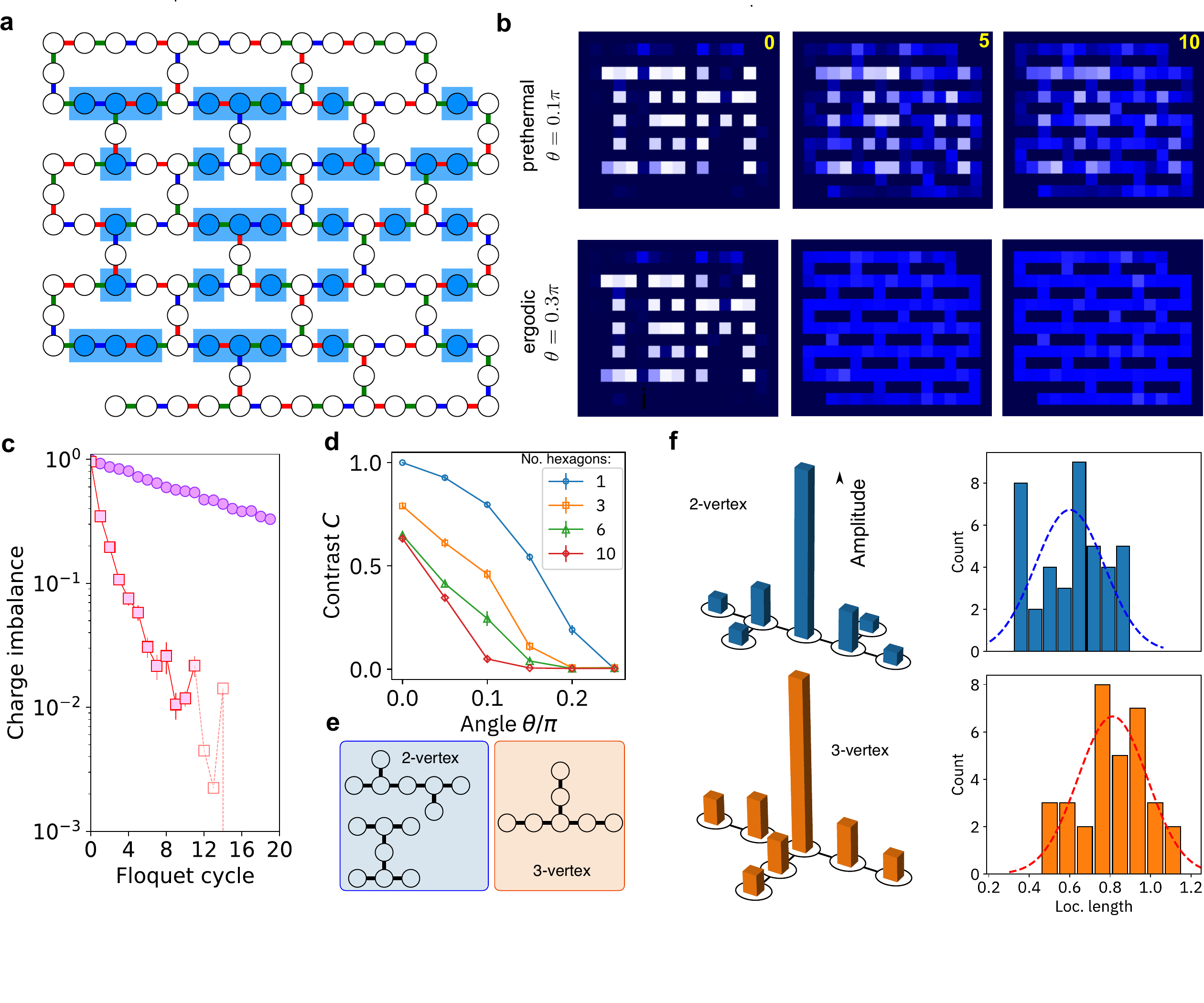}
\caption{\textbf{Benchmark signatures and uncovering LIOM operators in 2D.} 
\textbf{\arc} 
The 124 qubits (circles) of the heavy-hexagonal \texttt{ibm\_washington} processor used in the 2D experiments. Edge colors (blue, green, red) signify time order of two-qubit gates in a Floquet cycle. Arbitrary initial-excitation patterns can be encoded: blue fill for $|1\rangle$  and white for $|0\rangle$ initialization.
\textbf{\brc} 
Color-coded visualization of single-qubit excitations (white represents 1, dark blue 0) across the entire device for Floquet cycles $m=0$, 5, and 10, at two interactions: $\theta=0.1\pi$ (top) and $\theta=0.3\pi$ (bottom). Pattern diffusion at $\theta=0.1\pi$ suggests prethermalization.
\textbf{\crc} 
Spin imbalance for an initial antiferromagnetic pattern. Bootstrapped errors are depicted as vertical error bars positioned behind the markers. For $\theta=0.3\pi$ (red), signal is vanishing beyond cycle 10. 
For $\theta=0.1\pi$ (purple), gradual relaxation is observed without an evident plateau.
\textbf{\drc} OPDM gap contrast versus interaction strength~$\theta$ across system sizes measured in number of hexagon rings (see Supplementary Information ). Clear signature of MBL regime is not evident. Bootstrapped errors are depicted as vertical error bars positioned behind the markers.
\textbf{\erc} In 2D, LIOMs have two unique support topologies characterized by a central vertex with either two or three connections. Average LIOM operator weights in the bulk for each topology are shown (representing a total of 124 LIOMs).
\textbf{\frc} Imprecision histograms for all bulk two-vertex (top) and three-vertex (bottom) prethermal LIOMs. The shaded region indicates imprecision that exceeds that of the initial-guess operator.}
\label{fig:2dsystem}
\end{figure*}

To understand scaling with system size~$n$, in the right panel of Fig.~\ref{fig2:1dlocalization}\dr, we plot the order parameter~$\delta \times n$. 
With increasing system size, this quantity should theoretically diverge in the MBL regime but vanish in the ergodic regime, where thermalization leads to a featureless `infinite-temperature' state, and a cross-over between the two should be present. Experimentally, we observe behaviour that is consistent with the critical angle~$\theta_c \approx 0.16\pi$ obtained from the exact diagonalization.

The spectral discontinuity $\delta$ is sensitive to experimental noise, which could also close it. As a more noise-robust order parameter, we introduce the gap contrast~$C \coloneqq 1-\rho_c/\rho_p$, where $\rho_c$ is the density of orbitals (i.e. average number of orbitals per range of occupation values) at the center of the spectrum and $\rho_p$ is the value at the density peaks (see Fig.~\ref{fig2:1dlocalization}\er). We introduce this quantity here to observe the spectral signatures of the MBL even when the noise effects transform the spectral gap into a `soft gap', where the density of states is low but non-vanishing. The contrast~$C$ is maximized for noiseless MBL dynamics ($C=1$) and vanishes for ergodic dynamics ($C=0$). 
Our experimental observations in Fig.~\ref{fig2:1dlocalization}\er~ align with the expected behavior. For systems larger than those presented in this study, the restricted propagation of information in one dimension prevents full OPDM equilibration and, consequently, undermines the identification of the crossover position.\\

\textbf{Local integrals of motion (LIOMs) in 1D.} To demonstrate the capability of quantum hardware for studying disordered quantum systems, we develop an experimental protocol to construct a set of LIOMs~$L_1,L_2\dots,L_n$ from a system's time dynamics (recall Fig.~\ref{fig:schematics}). 
Here, we define a LIOM $L$ as a Hermitian operator that commutes with the Floquet unitary $[L,U_F] = 0$. To find each LIOM, we consider its decomposition $L = \sum_\mu a_\mu P_\mu$ over multi-qubit Pauli operators~$P_\mu$ indexed by $\mu\in[1,4^{n}]$, with real-valued weights $a_\mu$.
To find these weights, we use the fact that average inverse-time evolution of an arbitrary initial operator~$L^{(0)}$ is an exact LIOM~\cite{mierzejewski2015identifying,chandran2015constructing}. Thus, its Pauli weights can be estimated as
\begin{equation}\label{eq:fourier_transform}
a_\mu \simeq \frac{1}{2^n}\frac{1}{D+1}\sum_{d=0}^{D} \mathrm{Tr}\left[ U^{d}_F L^{(0)} U^{d\dag}_F P_\mu\right]\;,
\end{equation}
where $D \gg 1$ represents the number of Floquet cycles, and $\simeq$ denotes equality in the asymptotic limit $D \to \infty$. The optimal choice of depth $D$ is device-dependent and determined by the trade-off between noise and the finite sum error (see Supplementary Information~I.H).
 Estimating each $a_\mu$ on a quantum computer is scalable with a constant number of Floquet steps and measurements for a specific error target. Deriving the complete LIOMs, comprising the set of non-trivial $a_\mu$ coefficients up to an arbitrarily small constant error in operator norm, is more resource-intensive but remains independent of system size. Due to the local structure of LIOMs, identifying only the local Pauli components derived from local measurements offers a good approximation of the operators. It is worth noting that one may derive non-local properties of the system by further processing multiple LIOMs, e.g. constructing the eigenstates \cite{chandran2015spectral}.

We experimentally estimate~$a_\mu$ by measuring the time evolution of Pauli expectation values~$\langle P_\mu(d) \rangle$ \cite{elben2023randomized}; example data depicted in  Fig.~\ref{fig:1dLIOMs}\ar. To reduce the sampling overhead, we employ the limited size of lightcone that allows us to significantly reduce the number of samples in initial states and measurements, see Supplementary Information Section II.B. To lessen noise effects, this experiment is run to~$D = 10$ cycles.  
For each qubit~$i$, we select $L^{(0)}_i=Z_i$, the on-site $Z$ operator. We choose the parameter $\theta = 0.1\pi$ for several reasons. In this regime, the differences in coefficients for Pauli operators between the time average and the initial guess are large enough to be measured without significant error. Also, since the required averaging depth grows as $\propto\theta^{-1}$, the chosen value ensures that the noise error is comparable to the averaging error (see Supplementary Section I.H).  As depicted for $i=50$ in Fig.~\ref{fig:1dLIOMs}\ar, we find a combination of the numerous expectation values that yields an approximate constant of motion (blue dots) or LIOM. As the data show, the apparent chaos of expectation values hides order. 

To understand the effect of noise and partially validate the LIOM, we find the exact LIOM in the absence of noise for a 10-qubit subset of the 1D chain. The dashed line reports the noise-free LIOM value, aligning with the data up to some small variance, which we will shortly use to quantify the precision of the extracted LIOMs.

We next extract 104 different LIOMs for the same 1D chain with fixed disorder. One should expect that LIOMs must be linear combinations of Pauli operators with weights that decay exponentially with the distance from the center and the size of the Pauli operators. To illuminate this structure, we evaluate $k$-local qubit weights~$W_{k}(i) \coloneqq \sqrt{\sum_\mu a_\mu^2 \delta(i,k,\mu)}$, where $\delta(i,k,\mu)$ is equal to $1/k$ if~$P_\mu$ acts non-trivially on~$k$ qubits necessarily including qubit~$i$; otherwise, $\delta(i,k,\mu)=0$. Figure~\ref{fig:1dLIOMs}\br~ shows the average weights $W_{k}(i-i_0)$ for bulk LIOMs generated from Pauli-Z operators on qubits $i_0\in[6,97]$, along with the total weights~$W(i) = \sqrt{\sum_k W^2_k(i)}$.  The data clearly show the striking exponential localization of the LIOMs in both distance $|i-i_0|$ from their centers~$i_0$ and higher-weight operator locality~$k$.

How good is the average reconstruction? In Fig.~\ref{fig:1dLIOMs}\crf, we compare the experimental data (dots) with noise-free simulations (solid lines) obtained from brute-force simulations of a subsystem within the physical lightcone. 
The two agree notably across several orders of magnitude.

To gain further insight, we also evaluate an average LIOM characterized by normalized weights~$w_\mu = \operatorname{\mathbb E} [a^2_{\mu}]$ averaged over the bulk LIOMs when we align their centers.
We depict dominant Pauli contributions in Fig.~\ref{fig:1dLIOMs}\dr, which evidently  originate primarily from a few local Pauli operators. The central Pauli-Z dominates  with $w_Z = 0.71$---an observation in alignment with the robustness of the spin imbalance of the antiferromagnetic order observed in the benchmark experiments in Fig.~\ref{fig2:1dlocalization}\ar. Additionally, the central Paulis $X$ and $Y$ are also partially conserved, albeit with significantly lower average residual expectations of $w_X = 0.021$ and $w_Y = 0.016$.

Figure~\ref{fig:1dLIOMs} demonstrates the impact of experimental noise on the extracted localization by contrasting the localization lengths of LIOM operators in noisy experiments against noise-free simulations. The figure shows a generally mild broadening in experimental LIOMs, with a few notable outliers. Further precision analysis of these LIOMs is detailed in the Supplementary Materials II.D.\\

\textbf{Benchmarks in 2D.} While the experimental protocols we presented in 1D straightforwardly generalize to 2D (see processor topology in Fig.~\ref{fig:2dsystem}\ar), the dynamics are far more complex. Owing to the higher connectivity, information spreads faster in the system, which allows all-to-all communication in our experiments achieved in 15 Floquet cycles for the furthest pairs.

We prepare the spin polarization pattern depicted in Fig.~\ref{fig:2dsystem}\ar~and record its time evolution (see Fig.~\ref{fig:2dsystem}\br). Akin to the 1D case, for~$\theta=0.3\pi$, the pattern rapidly scrambles, suggesting system thermalization. In contrast, for~$\theta=0.1\pi$, the pattern retains a higher level of distinguishability, pointing to slower, prethermal dynamics. Using smaller values of $\theta$ would result in slower dynamics for the same noise rate, leading to significantly higher errors, and is therefore not done in this study.
Fig.~\ref{fig:2dsystem}\crf~shows the corresponding spin imbalance dynamics for an antiferromagnetic initial state  (see Supplementary Information II.B).
The imbalance decays faster than in 1D, and for $\theta=0.1\pi$ (top trace), no plateau associated with MBL is evident within the observation timescale. Nevertheless, the slow relaxation timescale, which is about six times slower than in the ergodic regime, allows us to treat it as a prethermal regime. 

In Fig.~\ref{fig:2dsystem}\dr, we report the OPDM gap contrast ratio~$C$, a signature of the MBL anomaly, with system size in number of 2D hexagons.  We observe strong system-size effects at the few-hexagon level, a clear closing of the soft gap for larger~$\theta$ and some reduction for smaller~$\theta$. 
Surprisingly, a clear signature of MBL is not evident.
Noise has a more significant effect on~$C$ in 2D, hence one should interpret the results with care. The errors of LIOMs in 2D are shown in Supplementary Information II.E.

\textbf{LIOMs in 2D.}
The LIOM extraction protocol generalizes straightforwardly to different system sizes and levels of connectivity. However, in the 2D bulk, possible pre-thermal LIOMs can exhibit two distinct support patterns centered around vertices with either 2 or 3 connections. Fig.~\ref{fig:2dsystem}\er~ illustrates the average extracted LIOM in the bulk, represented by the average operator weight~$W(i)$ plotted over spatial sites~$i$, with $i=0$ denoting a LIOM center. These LIOMs display geometric symmetries, and as anticipated, the two-vertex LIOMs exhibit shorter localization lengths.

\textbf{Finite-time analysis.}
In the computationally universal system we analyze, the observables we identify by finite-time analysis, while precise, are not guaranteed to be conserved indefinitely \cite{shiraishi2021undecidability}. Therefore, some caution is advised in interpreting them as indicators of the system's ultimate integrability. It is possible that the system might exhibit near-integrable behavior temporarily before ultimately transitioning to chaos. 
In a future experiment, increasing~$D$ [Eq.~\eqref{eq:fourier_transform}] and observing  the presence of ``fat'' tails or a breakdown in the LIOM convergence can serve as a  signature of an MBL instability.

\textbf{Discussion}. In conclusion, our protocol effectively identifies conserved operators in large-scale disordered many-body systems, showing scalability independent of system size and notable noise resilience, and can be used to readily extract key physical features such as localization lengths. 
Validating our protocol through control experiments in a simulable 1D regime enhances its credibility and underscores progress in addressing device noise for our 2D results.
Notably, we detect indications of quasi-integrability at short depths in 2D, along with the presence of a complete set of prethermal LIOMs.

While the current work may not directly tackle the most pressing physics questions surrounding the MBL phase, future advances in error mitigation and instrument performance \cite{Kim2023utility,liao2023machine,majumdar2023best} may enable the use of LIOM measurements to deepen our understanding of fundamental phenomena. These include disordered dynamics in higher dimensions \cite{choi2016exploring}, avalanche instability \cite{morningstar2022avalanches,sels2022avalanches,wdroeck2017stability}, and topological MBL \cite{wahl2020}. Our experiment suggests that future research should be able to explore such fundamental questions, shedding light on the physics of many-body thermalization and the stability of non-equilibrium quantum phases through the use of quantum computers. An interesting technical  extension of this work would be, instead of accessing the mathematical structure of LIOMs, to construct an oracle that implements LIOMs as operators. The simplest such oracle would be a direct implementation of the operation in Eq.~\eqref{eq:fourier_transform} by applying a random number of unitary Floquet layers, a local operator, and the inverse of the unitary evolution.

\textbf{Data availability} The figure data generated in this study along with the indivdual Pauli expectation values have been deposited in the Figshare database \cite{figshare_data}.

\textbf{Acknowledgements}. We thank Abhinav Deshpande, Oliver Dial,  Andrew Eddins, Daniel Egger, Bryce Fuller, Jim Garrison, Dominik Hahn, Abhinav Kandala, Will Kirby, David Layden, Haoran Liao, David Luitz, Swarnadeep Majumder, Antonio Mezzacapo, Toma\v{z} Prosen, James Raftery, Dries Sels, Kristan Temme, Maurits Tepaske, and Ken Wei for valuable discussions. We also thank the entire IBM Quantum team for developing and supporting the infrastructure necessary to run these large-scale quantum computations. We acknowledge the usage of the Qiskit Prototype for zero-noise extrapolation. O.S., A.S., and Z.M acknowledge support by grant NSF PHY-2309135 to the Kavli Institute for Theoretical Physics (KITP).

\textbf{Author Contributions.} O.S. designed the experiment and developed the theoretical calculations and numerical simulations. D.W., H.Z., and N.H. developed the application software and performed the experiment. Z.M. helped to guide and interpret the experiment. A.S. and R.M. contributed to the theoretical analysis. The manuscript was written by O.S., D.W., and Z.M. All authors provided suggestions for the experiment, discussed the results, and contributed to the manuscript.

\let\oldaddcontentsline\addcontentsline
\renewcommand{\addcontentsline}[3]{}
\bibliography{references}

\let\addcontentsline\oldaddcontentsline

\clearpage
\pagebreak

\setcounter{page}{1}
\setcounter{equation}{0}
\setcounter{figure}{0}
\renewcommand{\theequation}{S.\arabic{equation}}
\renewcommand{\thefigure}{S\arabic{figure}}
\renewcommand*{\thepage}{S\arabic{page}}

\onecolumngrid

\begin{center}
{\large \textbf{Supplementary Materials for \\ ``\ourtitle"}}\\
\vspace{0.5cm}
Oles Shtanko$^1$, Derek S. Wang$^2$, Haimeng Zhang$^{2,3}$,\\
Nikhil~Harle$^{2,4}$, Alireza Seif$^2$, Ramis Movassagh$^{5}$, and Zlatko Minev$^2$\\
\vspace{0.25cm}
\textit{
$^1$\ibmalmaden\\
$^2$\ibmyorktown\\
$^3$\deeusc\\
$^4$\yale\\
$^5$\ibmcambridge
}
\end{center}

\tableofcontents

\vspace{1cm}

\twocolumngrid

\section{Theoretical analysis}

In this section, we provide a brief overview of our circuit model and demonstrate its ability to exhibit many-body localization (MBL). First, we outline the main features of the circuit model under study. We then present numerical observations of the MBL regime in one-dimensional circuits and study its dependence on circuit parameters. Our discussion then delves into a detailed analysis of the methods used to study MBL-type dynamics, including spin imbalance and one-particle density matrix. We conclude with specific details regarding the detection of local integrals of motion.

\subsection{The model}

This work presents an experimental study of a circuit model that exhibits the many-body localization (MBL) regime \cite{morningstar2022avalanches} (which, potentially, corresponds to the MBL phase in the thermodynamic limit \cite{basko2006metal,pal2010mbl}) and its transition to the ergodic regime. 
Circuit dynamics are generated by applying the same pattern of gates to each cycle. The output of the $d$th cycle is described by the state
\be
|\psi_d\rangle = U_F^d|\psi_0\rangle,
\ee
where $|\psi_0\rangle$ is the initial state and $U_F$ is the Floquet unitary \cite{zhang2016floquet,ponte2015floquetmbl} of a cycle. We consider the Floquet unitary composed of a number of layers,
\be\label{eqs:floquet_decomposition}
U_F = \prod _kU^{(k)}_F,
\ee
where $U_F^{(k)}$ is the layer unitary of the form
\be\label{eqs:floquet_layer_decomposition}
\begin{split}
 &U_F^{(k)} =  \prod_{i=1}^n P_i(\varphi_i) \prod_{(i,j)\in C_k} G_{ij}(\theta),\\
 &G_{ij}(\theta) = U_{i}(\theta)U_j(\theta)C^Z_{ij},
 \end{split}
\ee
where $P_i(\varphi_i)$ and $U_i(\theta)$ are single-qubit gates in Eq.~\eqref{eq:gates} in the main text applied to qubit $i$, $G_i(\theta)$ is a composite two-qubit gate, and $C^Z_{ij}$ is the controlled-Z gate applied to qubits $i$ and $j$. Here, $C_k$ is a list of qubit pairs undergoing a two-qubit unitary gate in layer $k$.

We first consider a one-dimensional brickwork circuit, where each Floquet cycle consists of two layers (see Fig.~\ref{figs:circuit_layouts}a). The first layer, represented as $C_1$, contains gates acting on qubit pairs of the form $(2i, 2i+1)$, where $i$ ranges from $0$ to $\lfloor n/2\rfloor-1$. The second layer, symbolized by $C_2$, contains gates acting on pairs of the form $(2i+1, 2i+2)$, where $i$ ranges from $0$ to $\lceil n/2\rceil-2$. In contrast, each Floquet cycle of a two-dimensional circuit on a heavy-hexagonal lattice consists of three unique layers. The qubit pairs involved in each layer are indicated by different colors in Fig.~\ref{fig:2dsystem}a, as also illustrated in Fig.~\ref{figs:circuit_layouts}b. A crucial point to note is that each layer is completed by successive disorder gates. The random angle $\phi_i\in[-\pi,\pi]$ of each disorder gate acting on the qubit $i$ remains constant over all cycles and layers throughout the experiment. This scenario is similar to static disorder in quantum spin chains and differs from dynamic disorder, which is essentially noise.

\begin{figure}[t!]
    \centering\includegraphics[width=0.45\textwidth]{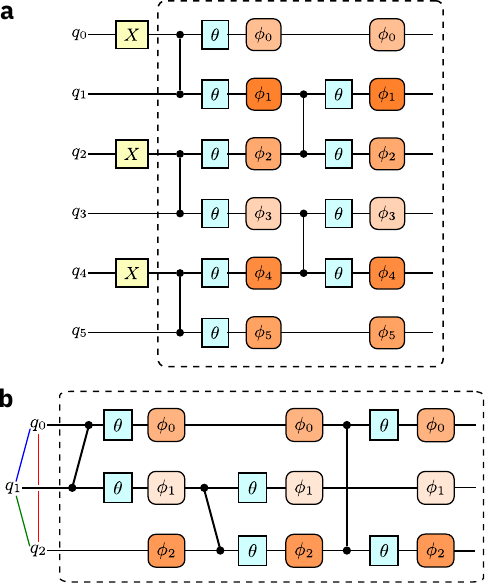}
    \caption{ 
    \textbf{Circuit model of Floquet dynamics}. 
    \textbf{a.} Gate schematic of initialization and one cycle of our Floquet model for an example of one-dimensional, 6-qubit lattice. 
    All qubits are  initialized to the target initial state using single-qubit  $X$ gates applied to the initial all-zero state. The boxed area represents a single Floquet cycle, which is iterated to advance the dynamics of the system.
    \textbf{b.} 
    Illustration of a single Floquet cycle in a two-dimensional lattice, showcasing three interconnected qubits. The colors of the connections on the left represent the sequential application of the two-qubit gates over their connectivity (i.e., graph coloring). Each layer consists  two-qubit $CZ$ gates (depicted as black dots connected by lines), $U(\theta)$ gates shown as green boxes, and disorder gates $P(\phi_i)$ depicted in orange.
    Visualizations were generated using \textit{Quantikz}.}
    \label{figs:circuit_layouts}
\end{figure}

We always initiate the qubits in a product state in the computational basis. Since each product state is a superposition of many eigenstates of the unitary $U_F$, this circuit leads to far-from-equilibrium dynamics. The nature of this dynamics depends on the values of the angle $\theta$. As we will show numerically in the following section, in one dimension there exists a critical value of $\theta_c$ such that for $\theta<\theta_c$ the spectrum of the unitary $U_F$ behaves similarly to a set of random phases, while the eigenstates are low-entangled, signaling the presence of many-body localization (MBL). In contrast, for $\theta>\theta_c$ we observe signatures of chaotic dynamics, including a spectrum resembling a random matrix and high eigenstate entanglement.

More details about the circuit are in the Section~\ref{sec:implementation}, including the implementation using native gates, simultaneous measurement of commuting observables, and error mitigation.

 \begin{figure*}[t!]
    \centering
    \includegraphics[width=0.9
    \textwidth]{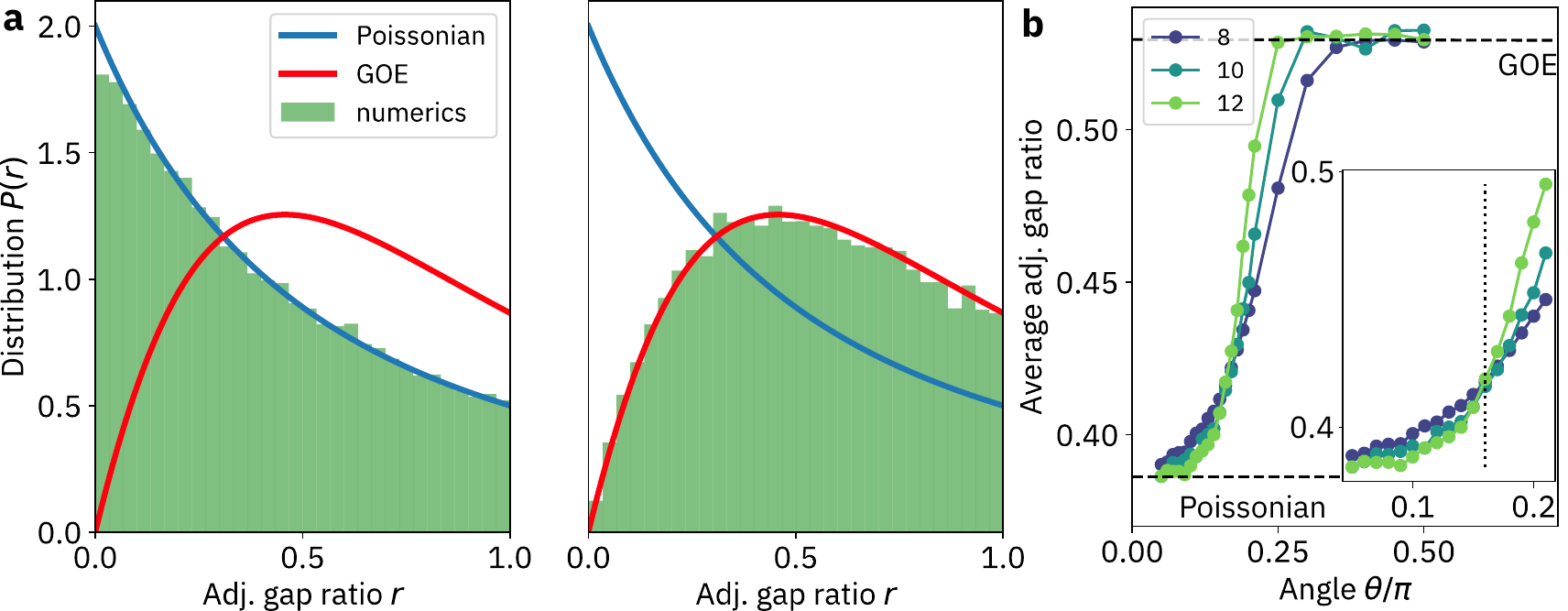}
    \caption{
\textbf{Quasienergy level statistics for the one-dimensional (1D) circuits.} \textbf{a.} 
The adjacent gap ratio distribution [see Eq.~\eqref{eqs:agr}] is depicted for the Floquet circuit at two different interaction $\theta$: $0.1\pi$ (left) and~$0.3\pi$ (right). 
The left plot corresponds to the many-body localized (MBL) dynamical regime, while the right plot corresponds to the ergodic regime. 
The green-shaded bars display the distribution obtained over $10^3$ random disorder realizations of  one-dimensional Floquet circuits with 12 qubits.
The solid blue and red curves represent fits of the numerically-obtained distributions using Poissonian and  Gaussian Orthogonal Ensemble (GOE) statistics, respectively. 
\textbf{b.} 
The transition from GOE to Poissonian statistics for different system sizes as shown by the average adjacent gap ratio for circuits containing from 8 to 12 qubits. The average is taken over $10^3$ individual circuits with different disorder realizations. The inset focuses on a crossover at $\theta \approx 0.16 \pi$ (indicated by a vertical dotted line).}
    \label{figs:levelstat}
 \end{figure*}
   \begin{figure*}[]\centering\includegraphics[width=0.63\textwidth]{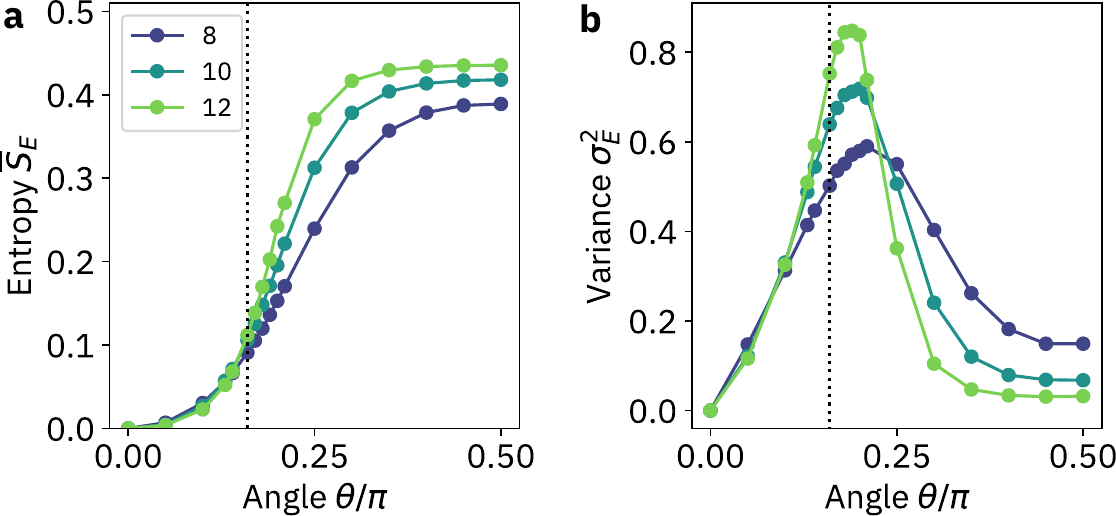}
    \caption{
      \textbf{Entanglement entropy for the one-dimensional (1D) circuits.}
      \textbf{a.} The average entanglement entropy $\bar S_E$ of eigenstates per qubit across varying system sizes, as obtained from $10^3$ disorder realizations for 8-qubit circuits and $10^2$ realizations for 10- and 12-qubit circuits. The system demonstrates a crossover at $\theta\approx0.16\pi$. \textbf{b.} The variance of entanglement entropy $\sigma_E^2$ among different eigenstates. This function has a pronounced peak that becomes sharper and converges to the critical point as system size increases.
    }
    \label{figs:enetanglement_entropy}
 \end{figure*}

\begin{figure*}[t!]
\centering
    \includegraphics[width=0.85\textwidth]{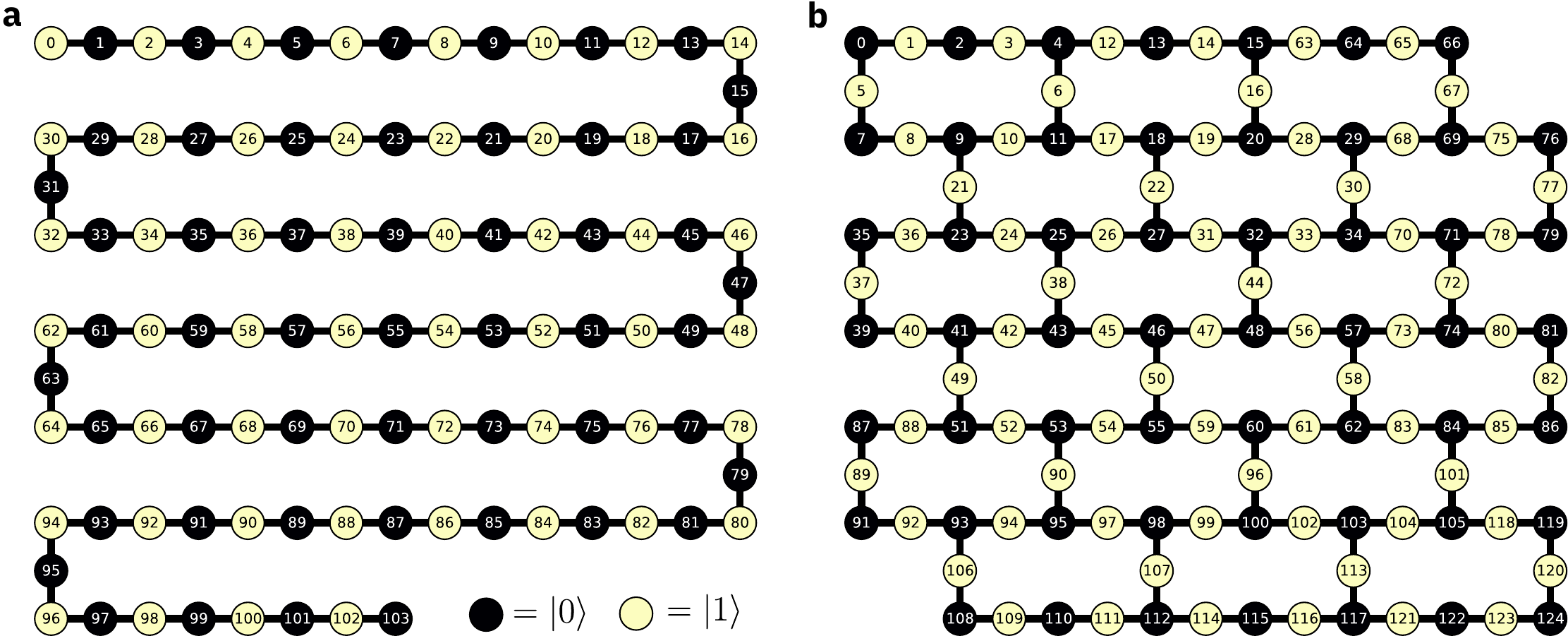}
\caption{
      \textbf{Antiferromagnetic order of the initial states.}
      The illustration displays the configuration for the initial states exhibiting antiferromagnetic order. 
      Black circles represent qubits initialized in $|0\rangle$, while yellow those initialized in~$|1\rangle$. 
      \textbf{a.} Configuration used for the 1D chain comprising 104 qubits. 
      \textbf{b.} Configuration of the 2D heavy-hexagon lattice, comprising 125 qubits.
    }
    \label{figs:cdw_state}
 \end{figure*}
 
\subsection{Evidence of MBL regime}
\label{mbl_evidence}

Here, we analyze the proposed model and demonstrate the existence of the MBL regime through numerical simulations performed on systems with 8 to 12 qubits. 

A crucial characteristic of the MBL phase can be identified in the spectrum of the Floquet unitary. To show this, we examine the spectral decomposition
\be
U_F|\alpha\rangle = \exp(i\lambda_\alpha)|\alpha\rangle,
\ee
where $\lambda_\alpha \in [-\pi,\pi]$ are quasienergies and $|\alpha\rangle$ are corresponding eigenstates. Assuming that $\lambda_\alpha$ are set in increasing order, we define the spectral gaps as
$
\delta \lambda_\alpha = \lambda_{\alpha+1}-\lambda_\alpha.
$
The main object we are interested in is adjacent gap ratios \cite{pal2010mbl,zhang2016floquet} that can be defined as
\be\label{eqs:agr}
r_\alpha = \frac{\min(\delta\lambda_\alpha,\delta\lambda_{\alpha+1})}{\max(\delta\lambda_\alpha,\delta\lambda_{\alpha+1})}.
\ee
The statistical distribution of $r_\alpha$ in the large system limit, denoted by $P(r_\alpha)$, has different values in the ergodic and MBL phases. In particular, the MBL phase has Poisson level statistics, while the ergodic phase has the statistics of the Gaussian orthogonal ensemble (GOE) of matrices:
\be
\begin{split}
&P(r) = \frac{2}{(1+r)^2} \quad {\rm (Poisson)},\\ &P(r) = \frac{27}{4}\frac{r+r^2}{(1+r+r^2)^{5/2}} \quad {\rm (GOE)}.
\end{split}
\ee
The difference between these distributions can be used to distinguish the two regimes. The associated order parameter is the mean spectral gap ratio $\overline{r} = 2^{-n}\sum_\alpha r_\alpha$. In the MBL phase, we expect this mean to reach the Poisson value $\overline{r} \approx 0.39$, while the GOE distribution for the ergodic phase has $\overline{r}\approx 0.53$ \cite{pal2010mbl}.

The spectral properties of our circuit are shown in Fig.~\ref{figs:levelstat}a. Here we see that for $\theta = 0.1\pi$ the level statistics almost perfectly follow the Poisson distribution. In contrast, for $\theta = 0.3\pi$ it changes to the GOE distribution. This transition is illustrated for several circuit sizes in Fig.~\ref{figs:levelstat}b. The curves show a strong crossover at $\theta_c\approx 0.16\pi$, indicating the presence of a transition between different level statistics. While the fate of this transition in the thermodynamic limit remains an open question \cite{wdroeck2017stability,morningstar2022avalanches,sels2022avalanches}, this critical value clearly distinguishes two dynamical regimes for small systems, which we refer to below as the ``MBL regime'' and the ``ergodic regime''.

Another measure of localization is the bipartite entanglement entropy of the eigenstates of the Floquet unitary. To define it, we divide the system into two equal parts and label them as subsystem $A$ and subsystem $B$. Then we can represent each eigenstate $|\alpha\>$ using Schmidt decomposition
\be
|\alpha\> = \sum_\mu \eta_{\alpha\mu} |\psi^A_{\mu\alpha}\>\otimes|\psi^B_{\mu\alpha}\>,
\ee
where $|\psi^A_{\mu\alpha}\>$ and $|\psi^B_{\mu\alpha}\>$ are some basis states of the subsystems $A$ and $B$, respectively. The entropy of bipartite entanglement is defined as
\be
S_\alpha := -\sum_{\mu} \eta^2_{\alpha\mu} \log_2\eta^2_{\alpha\mu}.
\ee
As an associated order parameter we use the average entanglement entropy per qubit,
\be
\overline S_E := \frac{1}{2^nn}\sum_\alpha S_\alpha.
\ee
For the MBL phase, we expect the entanglement entropy to exhibit an area law, i.e. it remains constant independent of the number of qubits $n$ \cite{pal2010mbl}. This means that in the MBL phase it has scaling $\overline S_E = O(n^{-1})$. In contrast, in the ergodic phase Floquet dynamics thermalizes the subsystems to local infinite temperature. In this state the entropy saturates at the Page value,
\be
\overline S_E  \simeq \frac{1}{2}-\frac{1}{2n\ln 2 }+O\left(2^{-n}\right).
\ee
Thus it saturates to a constant in the large system limit.
Another test of the MBL phase transition is the variance of the entanglement distribution over eigenstates, i.e.
\be
\sigma^2_E := \frac{1}{2^n}\sum_\alpha(S_\alpha-\overline S)^2.
\ee 
The maximum of the variance $\sigma_S$ indicates the point of the phase transition for $n\gg1$.

The entanglement entropy for the Floquet unitary is depicted in Fig.~\ref{figs:enetanglement_entropy}. Similar to the level statistics, the entropy per qubit exhibits distinct behaviors on opposite sides of the critical point $\theta_c\approx 0.16\pi$ as the system size increases. Notably, in the MBL regime, its value displays a gradual decline with increasing of the system size, whereas in the ergodic regime, it ascends towards the Page value. The entropy variance also has a pronounced peak near the transition point, which lends further support to the existence of an MBL regime and its transition into ergodic behavior in our system.

We do not present the evaluation of level statistics or entanglement entropy for two-dimensional systems as it technically challenging due to the exponential overhead in classical complexity. Consequently, it remains uncertain whether the regime we observe for low $\theta$ values exhibits characteristic features of MBL. We nonetheless refer to it as the ``prethermal regime" for simplicity, predicated on the assumption that localization signatures should persist in small two-dimensional systems at shallow depths. Our study does, however, highlight a significant disparity between the dynamics of one- and two-dimensional systems when viewed from the lens of the spectrum of the one-particle density matrix, as depicted in Fig.~\ref{fig2:1dlocalization}d and Fig.~\ref{fig:2dsystem}d (definition provided below). It is worth noting that this observed effect could potentially be a consequence of noise -- a hypothesis that can be confirmed in future studies with higher-fidelity quantum hardware.
 
\subsection{Spin imbalance}
\label{CDW_section}

Unlike level statistics or entanglement entropy, the spin imbalance -- defined in the following paragraph -- is a localization measure that can be efficiently determined on real hardware without the exponential sampling overhead. Its detection involves the preparation of a structured initial state and the measurement of a certain parameter that describes its local order. In the MBL regime, the order parameter is expected to never decrease below a certain value, since the local memory of the initial state persists over time. Conversely, in the ergodic phase, the local order dissipates rapidly, rendering the initial state irretrievable.

A commonly used ordered state is the antiferromagnetic ordered state (previous works also use the terms ``antiferromagnetic" or ``N\'{e}el state"). Its layout in one dimension is illustrated in Fig.~\ref{figs:cdw_state}a. To define the order parameter, we group all qubits into set $S_0$, which contains all qubits initialized in state $|0\>$, and set $S_1$ containing the qubits initialized in state $|1\>$. Then we define the spin imbalance at Floquet cycle $d$ as
\be\label{eqs:imbalance}
\mathcal I_d = \frac{n_{d1}-n_{d0}}{n_{d1}+n_{d0}},
\ee
where the average population $n_{dk}$ of the qubits in the set $S_k$ with size $N_k$ takes the form
\be
n_{dk} = \frac 12-\frac1{2N_k}\sum_{i\in S_k} \<\psi_d|Z_i|\psi_d\>, \quad |\psi_d\> = U_F^d|\psi_0\>.
\ee 
In two dimension
s, the choice of the ordered state has more freedom. We choose the antiferromagnetic state shown in Fig.~\ref{figs:cdw_state}b, where the nearest neighbors are always in different states. Note, however, that this configuration has a larger number $N_1$ of $|1\>$ states than the number $N_0$ of $|0\>$ states. 

The spin imbalance renormalized to compensate for the effect of noise is shown in Figs.~\ref{fig2:1dlocalization}b and \ref{fig:2dsystem}c. The details of the renormalization and further experimental data are given in Sec.~\ref{more_data} below.

\subsection{One-particle density matrices (OPDMs)}

Here, we study the one-particle density matrix (OPDM), an $n\times n$ matrix with its elements defined for an arbitrary multi-qubit state $|\psi\>$ as
\be\label{eq:opdm}
\rho_{ij}= \<\psi|a_i^\dag a_j|\psi\> - \<\psi|a_i^\dag|\psi\>\<\psi|a_j|\psi\>,
\ee
where we denote the Fock operator for hardcore bosons $a_i = \frac{1}{2}(X_i-iY_i)$.

We can study the spectrum of the OPDM through the eigenproblem
\be
\rho|\varphi_k\> = \nu_k|\varphi_k\>,
\ee
where $|\varphi_k\>$ are natural orbitals and $\nu_k$ are ordered occupation numbers, $\nu_{k+1}\geq \nu_k$. In the case of a Gaussian state of fermions or hardcore bosons, the OPDM offers a comprehensive description of the system. The state of the system can be seen as a collection of one-particle orbitals $|\varphi_k\>$ occupied with probabilities $\nu_k\in\{0,1\}$. For other many-body states, this description extends to a probabilistic mixture of quasiparticle states occupying orbitals with probabilities $\nu_k$ \cite{bera2015opdm,bera2017opdm}. The conservation of occupation numbers serves as a robust indicator of the stability of quasiparticles \cite{mahan2000many}.

We evaluate OPDM for late-time state $|\psi_D\>$ with the initial state of the system $|\psi_0\>$ set to the CDW state described in the previous section. As an order parameter, we use the discontinuity parameter
\be\label{eq:discontinuity}
\delta = \left\<\nu_{N_0+1}-\nu_{N_0}\right\>_{\rm dis},
\ee
where $\<\dots\>_{\rm dis}$ is the disorder average. The the subscripts $N_0$ indicates number of qubits prepared in the zero state. In the initial CDW state, the value of the discontinuity is $\delta = 1$. For ergodic systems in both the thermodynamic and large-depth limits $D\to\infty$, the discontinuity vanishes, $\lim_{D\to\infty} \delta \to 0$, as the system loses all local memory of the initial state and the diagonalized OPDM becomes proportional to identity matrix. In contrast, in the MBL regime, the discontinuity saturates at a constant value $\lim_{D\to\infty} \delta >0$ due to the presence of stable quasiparticles. For finite systems and restricted depth $D$, we can use finite-size analysis of the discontinuity to detect the presence of a phase transition, as we show in Fig.~\ref{fig2:1dlocalization}c and d.

\subsection{Local integrals of motion (LIOMs)}

The central topic of our study are integrals of motion, which we denote by symbol $L$. In particular, we aim to find a set of real coefficients $\{a_\mu\}$ such that the operator
\be
L = \sum_{\mu=1}^{4^n-1} a_\mu P_\mu
\ee
is conserved, i.e. $[L,U_F] = 0$, where $P_\mu$ denote all possible generalized Pauli operators except the identity (e.g., for 4 qubits this set would include, for example, the operators like $XIZY$ or $IIXZ$ but not $IIII$). 

In a quantum system comprising $n$ qubits, there exist precisely $2^n$ independent integrals of motion. One possible choice of these operators is the ket-bra set $\{|\alpha\rangle\langle \alpha|\}$, where $|\alpha\rangle$ are the eigenstates of the Floquet unitary. It is crucial to highlight that these integrals of motion are nonlocal, signifying that they depend on Pauli matrices $P_\mu$ acting on most of the qubits. However, for certain integrable systems, sets of integrals of motion can be derived from a closed set of local conserved operators. As an example, a unitary classical spin system exhibits a set of Pauli operators $Z_k$ that commute with the unitary evolution operator. These operators can be summed and multiplied between each other to generate arbitrary sets of conserved operators.

 In this context, the MBL phase can be understood similarly to a classical spin system \cite{serbyn2013lioms,huse2014lioms}. In particular, there are $n$ local integrals of motion (LIOMs), similar to $Z_k$. Unlike classical systems, however, they are not diagonal and act on many qubits. For these LIOMs, the coefficients follow an empirical bound
\be\label{eq:localization}
|a_{\mu}|\leq \exp(-(|P_\mu|+d_\mu)/\xi),
\ee
where $|P_\mu|$ is a full support of the Pauli matrix $P_\mu$, $d_\mu$ is the distance between the central qubit of the LIOM and the center of the operator $P_\mu$, and $\xi$ is the localization length. By \textit{full support} here we mean the size of the smallest continuous subset $P_\mu$ acts on non-trivially. For example, in one dimension, both $IIYZZ$ and $IXIZI$ have support~3. By the center of the operator $P_\mu$ we mean the geometrical center that the subset is supported on, or qubits 3 and 2 (assuming zero-indexing), respectively, in the two previous examples.

Due the method's error and effect of noise, obtained LIOMs will be approximate, with deviation quantified as
\be\label{eqs:operator_error}
\epsilon := \frac{\|[U_F,L]\|_F}{2\|L\|_F}.
\ee
This quantity is bounded as $0\leq \epsilon\leq 1$ and vanishes when $[U_F,L] = 0$. The numerator and denominator of this expression are given by
\be
\begin{split}
&\|[U_F,L]\|_F^2 = 2^{n+2} \sum_{\mu\nu}a_{\mu}B_{\mu\nu}a_{\nu},\\
&\|L\|_F^2 = 2^n \sum_\mu a_{\mu}^2,
\end{split}
\ee
where we defined the matrix of the connected single cycle time correlation function for two Pauli operators as
\be\label{eqs:B_matrix}
B_{\mu\nu} = \frac{1}{2}\Bigl(-\frac{1}{2^n}\Tr(U_F^\dag P_\mu U_FP_\nu)+\delta_{\mu\nu}\Bigl).
\ee
To avoid confusion with other errors, we refer to $\epsilon$ as \textit{LIOM imprecision} throughout the main text and subsequent sections. Using the matrix $B$, we can express the LIOM imprecision as
\be
\epsilon = \Biggl[\frac{ \sum_{\mu\nu}a_{\mu} B_{\mu\nu}  a_{\nu}}{\sum_\mu a_\mu^2}\Biggl]^{1/2} \equiv \sqrt{\tilde a^TB\tilde a},
\ee
where $\tilde a =\{ \tilde a_\mu\}$ represents the vector of normalized coefficients $\tilde a_\mu = a_\mu/\sqrt{\sum_\mu a_\mu^2}$.

Let us evaluate the imprecision $\epsilon_0$ for the one-dimensional circuit using $L\equiv Z_i$, where $i\neq 1,n$. To do this, we note that
\be
\begin{split}
\epsilon_0 &= \frac{1}{2^{n/2+1}}\|[U_F,Z_i]\|_F\\
&=\frac{1}{2^{n/2+1}}\|C^Z_{i,i+s}U_i(\theta)Z_iU_i(\theta)C^Z_{i,i+s}\\
 &\qquad\qquad-P^\dag (\phi_i)U_i(\theta)Z_iU_i(\theta)P(\phi_i)\|_F,
\end{split}
\ee
where $s = +1$ for even $i$ and $s=-1$ for odd $i$. Here, we used the Floquet unitary decomposition in Eqs.~\eqref{eqs:floquet_decomposition} and \eqref{eqs:floquet_layer_decomposition} and the fact that the Frobenius norm is invariant to multiplication of the operator by unitary either from the left or from the right. The operator terms inside the norm can be expressed as
\be
\begin{split}
&C^Z_{i,i+s}U_i(\theta)Z_iU_i(\theta)C^Z_{i,i+s} = Z_i\cos\theta+X_iZ_{i+s}\sin\theta\\
&P^\dag (\phi_i)U_i(\theta)Z_iU_i(\theta)P(\phi_i) = Z_i\cos\theta+X_i\sin\theta\cos\phi_i\\
&\qquad\qquad\qquad\qquad\qquad\qquad\qquad\qquad+Y_i\sin\theta\sin\phi_i.
\end{split}
\ee
As the result, we have
\be
\begin{split}
\epsilon_0 &= \frac{\sin\theta}{2^{n/2+1}}\|X_iZ_{i+s}-X_i\cos\phi_i-Y_i\sin\phi_i\|_F = \frac{\sin\theta}{\sqrt{2}}.
\end{split}
\ee
This result does not depend on $\phi_i$ and thus on the disorder realization or the qubit position (except for the first and the last qubit we did not consider above). The value of $\epsilon_0$ is shown as a dashed red line in Fig.~\ref{fig:1dLIOMs}e for $\theta = 0.1\pi$.

The following sections describe the process of generating LIOMs. This process depends on an initial guess. The accuracy of this initial guess is not critical, but an improved estimate is expected to lead to better convergence. We then present two strategies for formulating this initial estimate: one based solely on theoretical analysis and the other based solely on experimental data. Finally, we examine the convergence of our approach and its limitations, in particular those induced by the presence of noise.

\subsection{Generating LIOMs via time averaging}
\label{averaging_method}

Given an initial guess $L^{(0)}$, we derive the corresponding approximate LIOM using the expression
\be\label{eqs:d_liom_approx}
L^{(D)} = \mathbb E_{d \leq D}\left[\mathcal E^d_F(L^{(0)})\right],
\ee
where the \textit{inverse-time} Floquet cycle map $\mathcal E_F(\cdot)$ and its average $\mathbb E_{d \leq D}$ over cycles $0$ to $D$ cycles are defined as
\be\label{eqs:fourier_channel}
\begin{split}
&\mathcal E_F(O) : = U_F O U_F^\dag,\\
&\mathbb E_{d \leq D}[O_d] := \frac{1}{D+1}\sum_{d=0}^{D}O_d,
\end{split}
\ee
where $\{O_d\}$ is an arbitrary sequence of $D$ operators.
Note that in the inverse time the order of the operator $U_F$ and the conjugate operator $U_F^\dagger$ is reversed compared to the forward time evolution of operators in the Heisenberg picture (or in the order corresponding to a density matrix). This order is chosen so that it is convenient to express $L^{(D)}$ below in terms of the evolution of the eigenstates of the operator $L^{(0)}$. This operator becomes an exact LIOM in the limit $D\to\infty$, i.e. $[L^{(\infty)},U_F] = 0$. Indeed, 
\be
U_F L^{(\infty)} U^\dag_F = \lim_{D\to\infty}\frac{1}{D+1}\sum_{d=0}^{D} U^{d+1}_F L^{(0)} U^{d+1\dag}_F = L^{(\infty)}.
\ee  
In the following section, we explore the convergence of this method for finite $D$ (see Proposition 1). 

The Pauli coefficients of the operator in Eq.~\eqref{eqs:d_liom_approx} are
\be\label{eqs:a_to_w}
\begin{split}
a_\mu &= \frac{1}{2^n}\Tr(P_\mu L^{(D)})\\
&= \frac{1}{2^n}\frac 1D\sum_{d=0}^{D-1}  \Tr( U^{d\dag }_F P_\mu U^d_F L^{(0)}) \equiv \frac 1D\sum_{d=0}^{D-1} W_d(P_\mu,L^{(0)}),\\
\end{split}
\ee
which leads us to Eq.~\eqref{eq:fourier_transform} in the main text in the limit $D\to\infty$. Here, we used the notation
\be
W_d(P_\mu,L^{(0)}) =\frac{1}{2^n} \Tr( U^{d\dag }_F P_\mu U^d_F L^{(0)}).
\ee

As the first step for evaluating the trace in Eq.~\eqref{eqs:a_to_w}, we use the spectral decomposition of the inital guess operator $L^{(0)}$ as
\be
L^{(0)} = \sum_{\alpha=1}^{2^n} p_\alpha|\psi_\alpha\>\<\psi_\alpha|,
\ee
where $p_\alpha$ and $|\psi_\alpha\>$ are corresponding eigenvalues and eigenstates. This expression leads us to
\be
\begin{split}
W_d(P_\mu,L^{(0)}) = \frac{1}{2^n}\sum_{\alpha=1}^{2^n} p_\alpha \<\psi_\alpha|U_F^{d\dag}P_\mu U^d_F|\psi_\alpha\>.
\end{split}
\ee
 If we choose $L^{(0)}$ to be an operator in the computational basis, one could approximate the coefficients using a sum over a limited subset of bitstrings $\mathcal S$, which size $|\mathcal S|$ can be small. (Although one could use Monte Carlo sampling of these bitstrings, in the next section, we propose a more efficient method that is also suitable for shallow circuits.) Thus we get
\be\label{eqs:corr_funct_approx}
W_d(P_\mu,L^{(0)}) \approx \frac{1}{|\mathcal S|}\sum_{\alpha\in|\mathcal S|} p_\alpha \<\psi_\alpha|U_F^{d\dag}P_\mu U^d_F|\psi_\alpha\>.
\ee
An advantage of this method is that it allows us to obtain the coefficients $a_\mu$ individually, without accessing the whole operator $L$. Therefore, we can learn some properties of $L$ without running out of classical memory to keep its full representation. This is useful when the system is close to the critical point where the correlation length $\xi$ in Eq.~\eqref{eq:localization} is large, leading to a significant number of relevant coefficients $a_\mu$ that grows exponentially with $\xi$.  Another advantage of this method is that the noise may partially self-average as a result of summing over many Floquet cycles in Eq.~\eqref{eqs:a_to_w}. 

In the literature, it is often considered a set of LIOMs that are orthonormal with respect to inner products and do not retain the properties of Pauli operators. Such operators are commonly called logical bits or ``l-bits". Our operators do not have such properties, but one could obtain the set of l-bits by classical postprocessing using a nonlinear transformation of the obtained LIOMS.

\subsection{Choosing the initial guess}
\label{initial_guess} 

In this section, we present and discuss several strategies for selecting operators that can serve as initial guesses. The first strategy is the straightforward selection of Pauli $Z$ operators. The second strategy requires an understanding of the $U_F$ structure of the Floquet unitary and aims at minimizing the norm of the commutator $[U_F,L]$. The third strategy extracts the initial estimate from experimental data and reduces the variance of the operator $L^{(0)}$ throughout its time evolution. Both theoretical and experimental optimization strategies are not perfect as standalone LIOM discovery methods due to their lack of scalability: they require a complete classical representation of the LIOM. However, we believe they are adequate for making an initial guess as we can limit our considerations to LIOMs with a support that meets available computational resources. In the following, we discuss each of these strategies in further detail. \\

\textbf{Strategy 1: a simple guess}. Since the system is strongly disordered in the z-direction, it is reasonable to assume that the operators $Z_i$ for $i=1,\dots,n$ represent a good option for the initial guess,
\be
L^{(0)} \in \{Z_i\}.
\ee
Due to the depth limitations ($D\sim 10$) given the existing level of noise (see next section), we found this choice to be the best at illustrating the capability of quantum hardware to converge to better LIOMs. This is why it is used in the main text. However, we also explored other options that could generate even better initial guess values. Below we list two such methods that may be useful for deeper circuits.\\ 

\textbf{Strategy 2: theoretical optimization}. Let us assume that we have an access to the full Floquet unitary. We are looking for our guess in the form
\be\label{eqs:L0}
L^{(0)} = \sum_{\mu\in Q} a_{\mu 0} P_\mu, \qquad \sum_\mu a^2_{\mu0} = 1,
\ee
where $Q$ is a certain subset of Pauli operators with number $q$ of $k$-local Pauli operators. It is also beneficial to choose only Pauli operators in the computational basis to reduce the computational resources required to obtain full LIOMs, as we describe further in the following section.

Next, our goal is to minimize the error
\be
\epsilon_0 := \frac{\|[U_F,L^{(0)}]\|_F}{2\|L^{(0)}\|_F} = \sqrt{a_0^TBa_0},
\ee
where $a_0 = \{a_{\mu0},0\dots 0\}$ is the vector where all coefficient for Pauli operators in $Q$ are replaced by $a_{\mu0}$, and the rest of elements are zero.
This can be done by solving the eigenproblem for the matrix in Eq.~\eqref{eqs:B_matrix},
\be
\sum_{\nu\in Q} B_{\mu\nu} a_{\nu0} = b_0  a_{\mu0}
\ee
for the smallest eigenvalues $b_0$. The minimum error is $\epsilon = \sqrt{b_0}$. Because vector $a_0$ has only $q$ non-zero entrees, we need evaluate only a block of size $q\times q$ corresponding to nonzero $a_{\mu0}$, not the entire matrix $B$.\\

\textbf{Strategy 3: experimental optimization}. Consider $\Psi = \{|\psi_0^{(k)}\>\}$ to be a set of initial states and $Q$ is the set of Pauli matrices supported on a certain subset of qubits. Using the decomposition in Eq.~\eqref{eqs:L0}, we rewrite the expectation value of the evolved initial guess as
\be
\begin{split}
\<\psi_0^{(k)}|U_F^{d\dag}L^{(0)}U_F^d|\psi_0^{(k)}\> = \sum_{\mu\in \mathcal S} a_{\mu0} f_{\mu kd},
\end{split}
\ee
where the coefficients $f_{\mu k m}$ are the expectation values for Pauli operator $P_\mu$ after $n$ cycles of evolution for the $k$-th initial state, i.e.
\be
f_{\mu kd} := \<\psi_0^{(k)}|U_F^{\dag d} P_\mu U_F^d |\psi_0^{(k)}\>.
\ee
Then, the coefficients for an approximate LIOM can be found by minimizing the variation
\be
\{a_{\mu0}\} = {\rm argmin}\, \sum_k{\rm Var}_{d\leq D}\Bigl( \sum_{\mu \in Q} a_{\mu0} f_{\mu k d}\Bigl),
\ee
where ${\rm Var}_{d\leq D} [O_d] := \mathbb E_{d \leq D} [O^2_d]-(E_{d \leq D}[O_d])^2$. In general, the maximum circuit depth $D$ and set size $|\Psi|$ must satisfy $D|\Psi| \geq q$ to avoid overfitting. Notably, this method can be implemented even for shallow circuits. Technically, for noiseless circuits, depth $D=1$ is sufficient if the set $\Psi$ is large enough.

The optimal values of $a_{\mu0}$ can be found by solving the spectral problem
\be
A_{\mu\nu}a_{\nu 0} = wa_{\mu 0},
\ee
for smallest possible $w$, where the matrix $A$ is generated as
\be
\begin{split}
A_{\mu\nu} = \sum_k\Bigl(\mathbb E_{d \leq D}& [f_{\mu kd}f_{\nu kd}]\\
&-\mathbb E_{d \leq D}[f_{\mu kd}] \mathbb E_{d'\leq D}[f_{\nu kd'}]\Bigl).
\end{split}
\ee

Both strategy 2 and strategy 3 are not scalable because they require classical storage and processing of \textit{all} $a_{\mu0}$ that define the LIOM. This number grows exponentially with the support of LIOMs and quickly becomes intractable, especially near the phase transition point. However, by truncating the support, these methods can be used to find a good initial guess.

\subsection{Convergence and noise}
\label{convergence_and_noise}

This section focuses on evaluating the performance of the proposed method under noise constraints. First, we postulate two essential properties of the studied quantum dynamics. These properties are then used to derive two error bounds: with respect to finite depth and with respect to noise. By combining these error estimates, we determine the optimal circuit depth for the experiment.

The first property concerns the conservation of local operators. In the MBL regime, the values of most of the local operators are partially conserved due to their overlap with LIOMs. A formal statement of memory preservation can be expressed by an inequality applied to local operators $O$,
\be\label{eqs:memory}
\Tr(U_F^{d\dag}OU_F^dO) \geq \lambda_d(O) \|O\|_F^2,
\ee
where $\|\cdot\|_F$ is the Frobenius norm, $\lambda_d(O) \geq 0$ is an operator-specific memory parameter that remains constant or decreases monotonically with increasing depth $d$.
In the MBL phase, we expect the parameter to remain nonzero, i.e. $\lim_{d\to\infty} \lambda_d(O)>0$. If the system is egodic, we expect $\lambda_d(O) \gtrsim \exp(-d/D_{\rm th})$ for local operators, where $D_{\rm th}$ is the thermalization time scale. The term ``prethermal" refers to a regime in which thermalization proceeds at a very slow rate, such that $D_{\rm th} \gg D$. In other words, the thermalization time scale is much larger than the number of observation cycles $D$.

The second property is localization. For all systems, including localized ones, there is a Lieb-Robinson bound
that limits the instantaneous propagation of information for large distances. For the MBL phase, this bound can be expressed in a stronger, ``logarithmic'' form \cite{burrell2007bounds},
\be\label{eqs:localization}
 \|[U_F^{d\dag}O_1U_F^d,O_2]\| \leq c\|O_1\|\|O_2\|d^\alpha e^{-{\rm dis}(O_1,O_2)/\xi_0},
\ee
where $\|\cdot\|$ is the operator norm, ${\rm dis}(O_1,O_2)$ is the distance between the operators $O_1$ and $O_2$, $\xi_0$ is the localization length, $c$ is a contant, and $\alpha>0$ is the exponent. For slowly thermalizing ergodic systems, we expect this inequality to hold empirically if we replace the constant localization length by a slowly changing cycle-dependent one, $\xi_0\to\xi_d$. We further assume that the localization length grows very slowly with cycle number $d$, so that its change remains negligible within $D_{\rm th}$ cycles.

Using the first property, we can show that as the maximum cycle depth $D$ grows, the method progressively converges to the target LIOM.

\begin{prop}[Depth convergence] \label{prop_conv}
For any Floquet dynamics satisfying Eq.~\eqref{eqs:memory}, the imprecision of the operator in Eq.~\eqref{eqs:d_liom_approx} is bounded as
\be\label{eqs:convergence}
\epsilon := \frac{\|[L^{(D)},U_F]\|_F}{2\|L^{(D)}\|_F} \leq \frac{1}{\sqrt{\lambda_D(L^{(0)})}D}.
\ee
\end{prop}

\begin{proof}
First, we rewrite the commutator as
\be
\begin{split}
[U_F,L^{(D)}] &= \Bigl(U_F L^{(D)} U^\dag_F - L^{(D)}\Bigl)U_F \\
&= \frac{1}{D}\Bigl(U_F^{d+1}L^{(0)}U^{d+1\dag}_F-L^{(0)}\Bigl)U_F.
\end{split}
\ee
Next, using the invariance of the Frobenius norm under rotations as well as the triangle inequality, we derive that
\be\label{eqs:norm_commut}
\begin{split}
\|[U_F,L^{(D)}]\|_F & =\frac{1}{D}\Bigl\|\mathbb E_{d \leq D}\Bigl(U_F^{d+1}L^{(0)}U^{ d+1\dag}_F-L^{(0)}\Bigl)U_F\Bigl\|_F \\
& = \frac{1}{D}\Bigl\|\mathbb E_{d \leq D} U_F^{ d+1}L^{(0)}U^{d+1\dag}_F-L^{(0)}\Bigl\|_F \\
& \leq \frac{2}{D}\|L^{(0)}\|_F.
\end{split}
\ee
Finally, we rewrite the norm of the operator $L^{(D)}$ using its definition as
\be\label{eqs:norm_operator_D}
\begin{split}
\|L^{(D)}\|_F&\equiv \sqrt{\Tr(L^{(D)2})} \\
&= \sqrt{\mathbb E_{d \leq D}\Tr\Bigl(U_F^d L^{(0)} U_F^{d\dag} L^{(0)}\Bigr)} \\
&\geq \sqrt{\lambda_D(L^{(0)})} \|L^{(0)}\|_F,
\end{split}
\ee
where in the last step we used Eq.~\eqref{eqs:memory}. By inserting Eq.~\eqref{eqs:norm_commut} and \eqref{eqs:norm_operator_D} into Eq.~\eqref{eqs:convergence} we get the desired result. This step completes our proof.
\end{proof}

Note that the right-hand side of the Eq.~\eqref{eqs:convergence} depends on $\lambda_D(L^{(0)})$. This implies that an accurate initial guess, denoted by $L^{(0)}$, can significantly improve the efficiency of the procedure. In particular, increasing the value of $\lambda_D(L^{(0)})$ will reduce the corresponding error. Thus, the selection of a better initial guess for $L^{(0)}$ is crucial to minimize the error and improve the overall accuracy and efficiency of the method. We also expect $\lambda_D(L^{(0)})$ to decrease exponentially with the localization length $\xi_0$. Thus, systems with smaller localization lengths should show significantly faster convergence.

Improvements made by increasing depth are limited by the effects of noise. Two primary categories of noise affect the system: systematic and stochastic. Systematic noise comes from calibration and readout errors. Predicting the influence of this type of noise is particularly challenging because it is hardware-dependent and can vary subtly over time. However, its impact is most noticeable for a few ``problematic" qubits and/or qubit coupling, see Fig.~\ref{fig:noiseerrorLIOMs}. For the majority of qubits, the main source of error is stochastic noise.

We consider a simplistic noise model where we assume that after each Floquet cycle the qubits are subjected to a single-qubit depolarizing channel. In this context, the unitary transformation $\mathcal E_F$ previously described in Eq.~\eqref{eqs:fourier_channel} must be replaced by its noisy version, which has the form
\be\label{eqs:noisy_dynamics}
\mathcal E'_F := \prod_{r=1}^{n} \Bigl((1-p)\hat I+\frac{p}{3}\sum_{ar}\Lambda_{ar}\Bigl) \mathcal E_F,
\ee
where $p$ is the error probability, $\hat I$ is the identity superoperator, and $\Lambda_{ar}$ is a single-qubit Pauli channel.
\be
\Lambda_{ar}(\cdot) := P_{ar}(\cdot)P_{ar},
\ee
and $P_{ar}\in\{X_r,Y_r,Z_r\}$ are single-qubit Pauli operators.

Consequently, the outcome obtained from the experimental data corresponds is the modified approximate LIOM $\tilde L^{(D)}$ that takes the form
\be\label{eqs:noisy_liom}
\tilde L^{(D)} = \mathbb E_{d \leq D}\mathcal E'^d_F(L^{(0)}).
\ee
To establish the effect of the noise, we prove the following Proposition.
\begin{prop}\label{noise_err_prop}
For any noiseless dynamics satisfying \eqref{eqs:localization}, the noisy dynamics in Eq.~\eqref{eqs:noisy_dynamics} has the action
\be\label{eqs:approximation}
\mathcal E'^d_F(L^{(0)}) = \mathcal E^d_F(L^{(0)})+\delta C_d,
\ee
where the correction operator $\delta C_d$ satisfies
\be
\frac{\|\delta C_d\|}{\| L^{(0)}\|}\leq cpd\sum_{r}\exp(-\Delta_r(L^{(0)})/\xi_d),
\ee
and $\Delta_r$ is the lattice distance between the qubit $r$ and the support of operator $L^{(0)}$.
\end{prop}

\begin{proof} We start by stochastically decomposing the noisy evolution superoperator as
\be
\begin{split}
\mathcal E_F'^d & = \sum_{\vec r\vec a\vec d} p_{\vec r\vec a\vec d}\mathcal E_F^{d-d_1}\Lambda_{a_1r_1}\mathcal E_F^{d_1-d_2}\Lambda_{a_2 r_2} \dots \Lambda_{a_kr_k}\mathcal E_F^{d_k}\\
& = \mathcal E_F^d \sum_{\vec r\vec a\vec d} p_{\vec r\vec a\vec d}\Lambda^{(d_1)}_{a_1r_1}\Lambda^{(d_2)}_{a_2 r_2} \dots \Lambda^{(d_k)}_{a_kr_k},
\end{split}
\ee
where $p_{\vec r\vec a\vec m} = (1-p)^{nd-k}(p/3)^{k}$ represents the probability of a combination of $0\leq k\leq nd$ single-qubit Pauli errors $\vec a = \{a_1, \dots,a_k\}$ involving qubits $\vec r = \{r_1,\dots, r_k\}$ at cycle depths $\vec d = \{d_1,\dots, d_k\}$. Here we have introduced the time-renormalized unitary error channels as
\be
\Lambda_{ar}^{(d)}(\cdot) := \mathcal E_F^{\dag d}\Lambda_{ar}\mathcal E_F^{d}= V_{ar}^{(d)}(\cdot)V_{ar}^{(d)\dag},
\ee
where we use the notation
\be
V_{a r}^{(d)} := U_F^{\dag d}P_{a r}U_F^{d}
\ee
and $P_{ar}\in \{X_{r},Y_{r},Z_{r}\}$ is single-qubit Pauli operator acting on qubit $r$.

\begin{figure}[t!]
\centering
\includegraphics[width=0.37\textwidth]{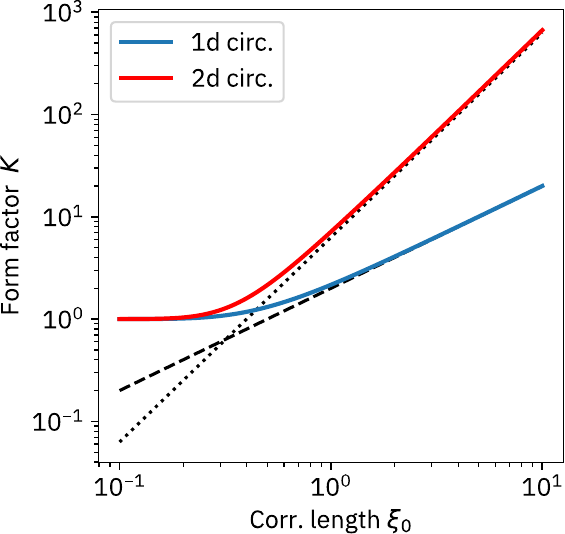}
\caption{\textbf{Noise form factor}. Form factor $K$ as a function of correlation length $\xi_0$ for $c=1$. The blue curve shows $K$ for infinite one-dimensional circuits, while the red curve shows a similar dependence for two-dimensional heavy-hexagonal circuits. The dashed curve shows the asymptotics $K = 2\xi_0$ and the dotted curve shows $K=2\pi\xi_0^2$. The support of the operator $L^{(0)}$ in the two-dimensional case is a qubit with three neighbors (the case of two-neighbor qubit is very similar and is not shown).}
\label{fig:noiseformfactor}
\end{figure}

\begin{figure*}[t!]
\centering
\includegraphics[width=1
\textwidth]{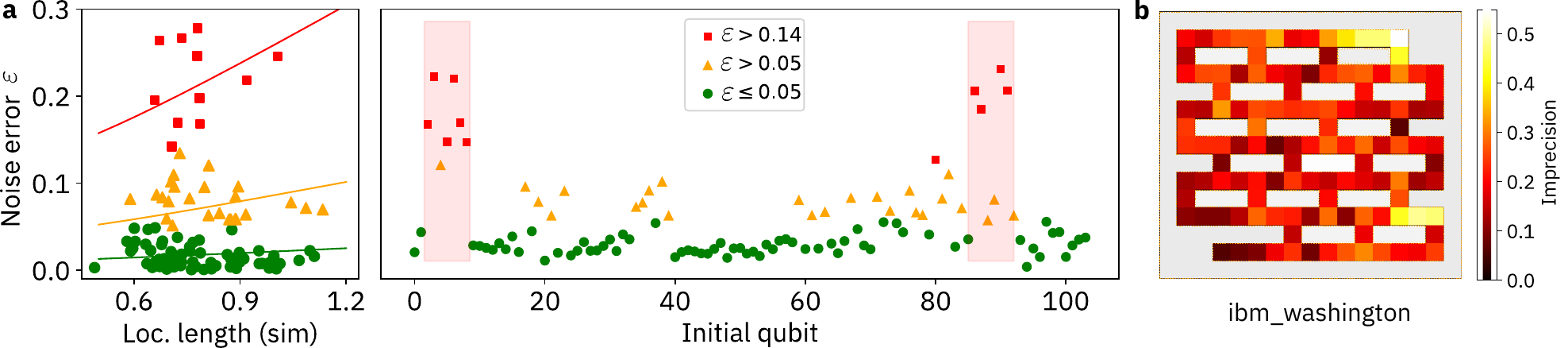}
\caption{
\textbf{Local integrals of motion (LIOMs) error due to noise.}
\textbf{a.} Noise error for 1D experiment. The plot shows the deviation of the coefficients $\vec a = \{a_\mu\}$ in Eq.~\eqref{eqs:a_to_w} obtained experimentally from the same coefficients $\vec a_{\rm sim}$ obtained by noiseless simulation, as quantified by the noise error $\varepsilon = \|\vec a-\vec a_{\rm sim}\|_2/\|\vec a_{\rm sim}\|_2$. For illustration, the LIOMs are divided into three groups: low precision (red, $\varepsilon>0.14$), medium precision (orange, $0.05<\varepsilon\leq 0.14$), and high precision (green, $\varepsilon\leq 0.05$). Left panel shows the distribution of LIOMs as a function of their localization length (derived from simulations). The panel shows that low- and high-precision LIOMs have similar localization lengths. Solid lines represent the theoretical estimate for $\varepsilon$ in Eq.~\eqref{eqs:optimal_params} for $p/\sqrt{\lambda_0} = 10^{-3}$ (green), $4\cdot 10^{-2}$ (orange) and $1.2\cdot 10^{-2}$ (red). The right panel shows the dependence of the error on the position of the initial guess $L^{(0)}$. This panel shows that the high errors are clustered around certain regions (shaded by red), suggesting that high errors are generated by a few ``bad" qubits or connections. \textbf{b.} LIOM imprecision for a 2D experiment. Each qubit $i$ is assigned a value of LIOM imprecision associated with the initial guess $L^{(0)} = Z_i$. Values are displayed in a 2D grid where neighboring cells correspond to neighboring qubits. As in 1D, anomalous values of imprecision are concentrated in a few regions. }
\label{fig:noiseerrorLIOMs}
\end{figure*}

Next, we rewrite the action of $d$ cycles of noisy evolution as  in Eq.~\eqref{eqs:approximation}, where we introduce the correction operator as
\be
\begin{split}
\delta C_m = \mathcal E_F^m\sum_{\vec r\vec a\vec m} p_{\vec r\vec a\vec m}\Bigl(\Lambda^{(d_1)}_{a_1r_1}\Lambda^{(d_2)}_{a_2 r_2} \dots \Lambda^{(d_k)}_{a_kr_k}(L^{(0)})-L^{(0)}\Bigl).
\end{split}
\ee
Taking the Frobenius norm of this operator and using the fact that it is invariant under unitary transformations $\mathcal E_F$, we get
\be
\begin{split}
\|\delta C\| &\leq \sum_{\vec r\vec a\vec d} p_{\vec r\vec a\vec d}\|\Lambda^{(d_1)}_{a_1r_1}\Lambda^{(d_2)}_{a_2 r_2} \dots \Lambda^{(d_k)}_{a_kr_k}(L^{(0)})-L^{(0)}\|_F.
\end{split}
\ee
Using the invariance of the Frobenius norm under unitary transformations, we get
\be
\begin{split}
\|\Lambda^{(d_1)}_{a_1r_1}\Lambda^{(d_2)}_{a_2 r_2} \dots &\Lambda^{(d_k)}_{a_kr_k}(L^{(0)})-L^{(0)}\|\\
&= \|[V^{(d_1)}_{a_1r_1}V^{(d_2)}_{a_2 r_2} \dots V^{(d_k)}_{a_kr_k},L^{(0)}]\|\\
& \leq \|[V^{(d_1)}_{a_1r_1},L^{(0)}]\|+\dots+\|[V^{(d_k)}_{a_kr_k},L^{(0)}]\|
\end{split}
\ee
Each of these summands can be expressed as
\be
\begin{split}
\|[V^{(d)}_{ar},L^{(0)}]\| = \|&[U_F^{\dag d}P_{a r}U_F^{ d},L^{(0)}]\| \\
&\leq c\|L^{(0)}\|d^\alpha\exp\bigl(-{\rm dis}(P_{ar},L^{(0)})/\xi_d\bigl)
\end{split}
\ee
where we used the condition in Eq.~\eqref{eqs:localization} with $\xi_d$ as a localization length. Therefore
\be
\begin{split}
\|\Lambda^{(d_1)}_{a_1r_1}\Lambda^{(d_2)}_{a_2 r_2} \dots& \Lambda^{(d_k)}_{a_kr_k}(L^{(0)})-L^{(0)}\| \\
 &\leq c\|L^{(0)}\|d^\alpha\sum_{i=1}^k \exp(-{\rm dis}(P_{a_ir_i},L^{(0)})/\xi_d),
 \end{split}
\ee
where we also used the inequality in Eq.~\eqref{eqs:norm_operator_D}. Using the fact that the noise is homogeneous, i.e. applies to all qubits with the same rate, we derive that
\be
\begin{split}
\|\delta C_d\| \leq c\|L^{(0)}\| pd^{\alpha+1} \sum_{r}\exp(-\Delta_r(L^{(0)})/\xi_d).
\end{split}
\ee
Thus, we arrive at our main result.
\end{proof}

For simplicity, we assume that the localization parameters remain constant over the time of our observation, i.e. $\lambda_m(L^{(0)}) \approx \lambda_0$ and $\xi_m\approx \xi_0$. Also, we introduce a form factor
\be\label{eqs:formfactor} K = c\sum_{r}\exp(-\Delta_{r,r_0}/\xi_0),
\ee
where $\Delta_{r,r'}$ is the distance between the qubits $r$ and $r'$. In essence, the form factor represents the degree to which the effect of single-qubit noise is multiplied by the interaction of the qubits for a given circuit geometry and localization length. The dependence of form factor $K$ on the localization length for one- and two-dimensional circuits we use is illustrated in Fig.~\ref{fig:noiseformfactor}.

Then the total error can be estimated as
\be\label{eqs:experimental_error}
\epsilon' = \frac{\|[U_F,\tilde L^{(D)}]\|_F}{2\|\tilde L^{(D)}\|_F}
\ee
where the ``noisy" LIOM is
\be
\tilde L^{(D)} = \frac{1}{D+1}\sum_{d=0}^D \mathcal E'^d_F(L^{(0)}) = L^{(D)} + \frac{1}{D+1}\sum_{d=0}^D \delta C_d,
\ee
and we defined $\delta C$ in Proposition~\ref{noise_err_prop}. Next, we take into account that the error grows monotoneously, $\|\delta C_{d}\|_F\leq \|\delta C_{d+1}\|_F$. Thus the total error is
\be
\epsilon'\leq \frac{\|[U_F,L^{(D)}]\|_F}{2\|\tilde L^{(D)}\|_F}+\frac{\|\delta C_D\|_F}{\|\tilde L^{(D)}\|_F},
\ee
Considering that $\|\tilde L^{(D)}\|_F \simeq \|L^{(D)}\|_F$ in the limit of small $p$, as well as the inequality in Eq.~\eqref{eqs:norm_operator_D}, we get
\be
\epsilon'\lesssim \frac{1}{\sqrt{\lambda_0}}\Bigl(\frac{1}{D}+\frac{\|\delta C_D\|_F}{\|L^{(0)}\|_F}\Bigl),
\ee
where we used the invariance of the Frobenius norm to the unitary transformation and Proposition~\ref{prop_conv}.
Next, we use the inequality connecting the spectral norm to the Frobenius norm, as well as the fact that $L^{(0)}$ is a Pauli operator, which leads us to
\be
\|\delta C_D\|_F \leq 2^{n/2}\|\delta C_D\|, \quad \|L^{(0)}\|_F = 2^{n/2},
\ee
where $n$ is the number of qubits. Using the expession for $\|\delta C_D\|$ from Proposition~\ref{noise_err_prop} and Eq.~\eqref{eqs:formfactor}, we arrive at the expression
\be
\epsilon'\lesssim \frac{1}{\sqrt{\lambda_0}}\Bigl(\frac 1D+pKD^{\alpha+1}\Bigl).
\ee
The optimal depth corresponding to the minimum error and the minimum error value itself are given by the expressions
\be\label{eqs:optimal_params}
\begin{split}
& D_{\rm opt} = \Bigl((1+\alpha)pK \Bigl)^{-\frac{1}{\alpha+2}},\\
& \epsilon'_{\rm min} \simeq\frac{\alpha+2}{\alpha+1} \frac{1}{\sqrt{\lambda_0}}\Bigl((1+\alpha) pK\Bigl)^{\frac{1}{\alpha+2}}.
\end{split}
\ee
Assuming $\alpha = 1$, and using realistic estimates for the gate noise $p\sim 10^{-3}-10^{-2}$ and the parameters $\sqrt{\lambda_0},K\sim 1$, we derive $D_{\rm opt}\sim 5-10$ and $\epsilon\sim 0.1$. These values are in agreement with the experimental results for the majority of LIOMs, as shown in Fig.~\ref{fig:noiseerrorLIOMs}. Furthermore, these results imply a polynomial convergence error of the method with increasing gate fidelity, denoted by $\epsilon_{\rm min} =O(p^{1/3})$.

\section{Experimental workflow}\label{sec:implementation}

\subsection{Hardware}

This section provides details of the experimental implementation. The quantum circuits are run on superconducting quantum processors \texttt{ibm\_kolkata} (27 transmon qubits) and \texttt{ibm\_washington} (127 qubits). The distribution of coherence times for \texttt{ibm\_washington} is shown in Fig.~\ref{fig:deviceerror}a. For this device, the median $T_1$ time of the qubits is 100 $\mu$s, while $T_2$ is 95 $\mu$s. 

The circuit design for both one and two dimensions involves nearest-neighbor coupling between qubits using a single controlled-Z gate per cycle. This gate is further decomposed into native gates, including a single two-qubit controlled-X gate and $SX$ and $RZ$ gates. The error of the controlled-X gate, with a median average of $1.15\%$ on \texttt{ibm\_washington}, is typically one or two orders of magnitude larger than that of a single-qubit gate, which has a median error of $\sim$0.03\%; see Fig.~\ref{fig:deviceerror}b. The provided median readout error is $1.47\%$. The provided gate errors were measured for a simultaneous run, but for potentially different order from our experiment. See also the gate fidelity study for linear chain for \texttt{ibm\_washington} in Ref.~\cite{McKay2023}.

\begin{figure*}[t!]
\centering
\includegraphics[width=0.85
\textwidth]{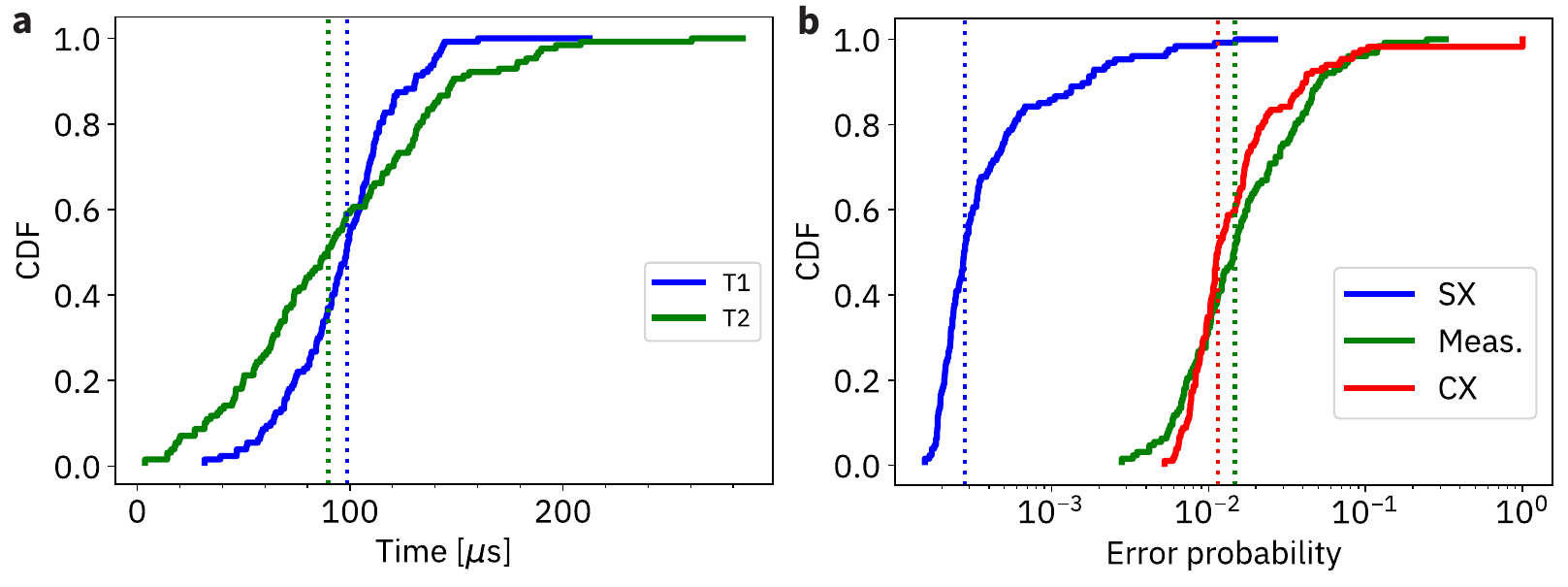}
\caption{ \textbf{Hardware error.} \textbf{a.} Cumulative Distribution Function for $T_1$ and $T_2$ times for \texttt{ibm\_washington.} The vertical lines show the mean times, which are 98.6 $\mu$s for $T_1$ and 89.8 $\mu$s for $T_2$. \textbf{b.} Instruction errors: $SX$ gate, $CX$ gate, (red) and measurement (green). As in the left panel, the vertical lines correspond to the medians for values of $2.81\cdot10^{-4}$ for the $SX$ gate, $1.15\cdot10^{-2}$ for the $CX$ gate, and $1.47\cdot10^{-2}$ for the measurement error.}
\label{fig:deviceerror}
\end{figure*}

 The circuits are submitted to the cloud-hosted devices via the Qiskit IBMQ provider, without exclusive access to device-level calibrations. In total, the results in the main text represent 35,364 raw circuits before error mitigation. Assuming a conservative sampling rate of 2 kHz as in Ref.~\cite{Kim2023} with 8,000 shots per circuit (including error suppression, such as twirling), the total quantum computation time for unmitigated results is $\sim$40 hours. As we discuss in the next section, the total overhead for error mitigation is twice the number of shots, so the mitigated results require a total of $\sim$80 hours of QPU runtime. Additionally, we use the shot budget for error mitigation techniques such as twirling, as detailed below.

\begin{figure*}[t]
\includegraphics[width=0.95\textwidth]{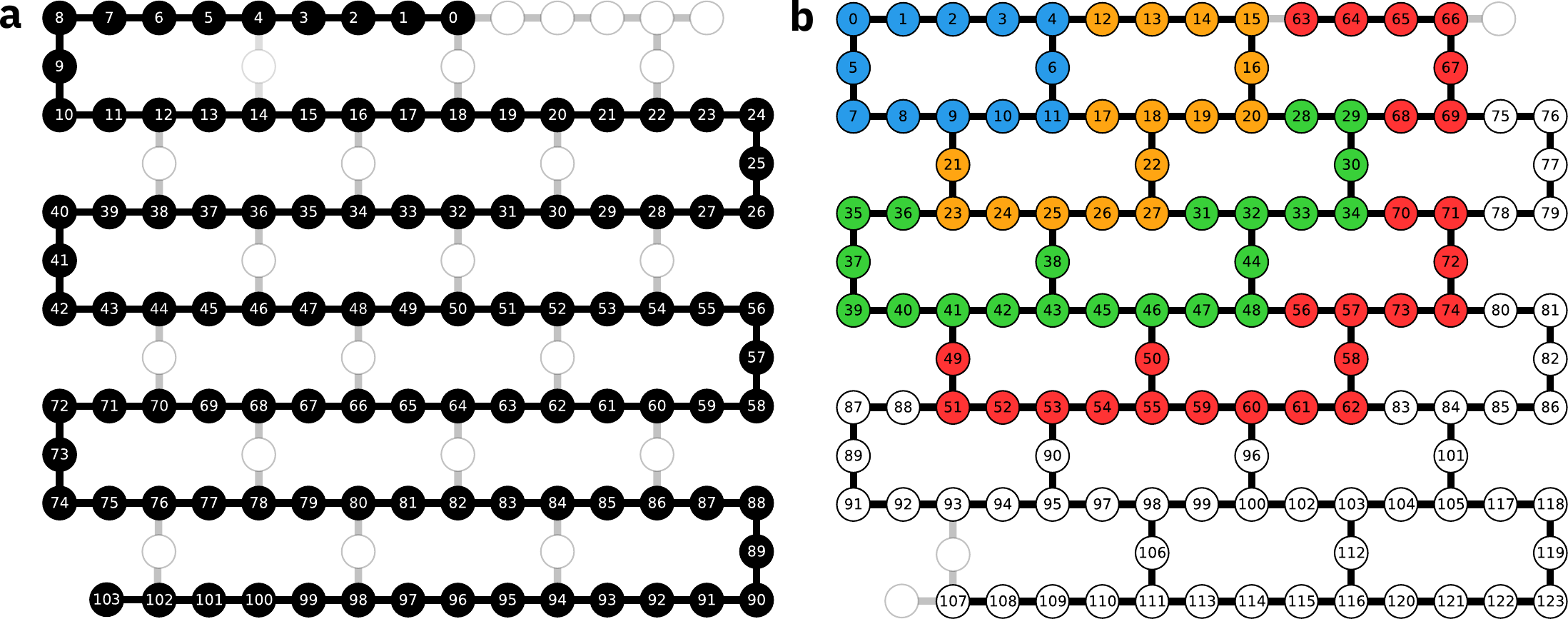}
\caption{
 \textbf{Mapping virtual to physical qubits in the quantum processor.} 
 The figures depict the topology of the 127-qubit quantum processors~\texttt{ibm\_washington} overlaid with virtual qubit embeddings used in our experiments. Qubits used in the virtual to physical mapping are shaded, colored, and numbered, while the remaining qubits are represented as open rings.
\textbf{a.} Schematic of the 104-qubit 1D chain used for studying spin imbalance.
\textbf{b.} Schematic of the 124-qubit 2D lattice. The various colors of the qubits illustrate the growth in system size for the OPDM experiments. Specifically, blue indicates the smallest system (1-hexagon), orange indicates the additional qubits for a 3-hexagon system, green represents the increase to a 6-hexagon system, and red denotes the additional qubits used for a 10-hexagon system.
}
\label{fig:logicalphysicalmapping}
\end{figure*}

\begin{figure}[t]
\centering
\includegraphics[width=0.35\textwidth]{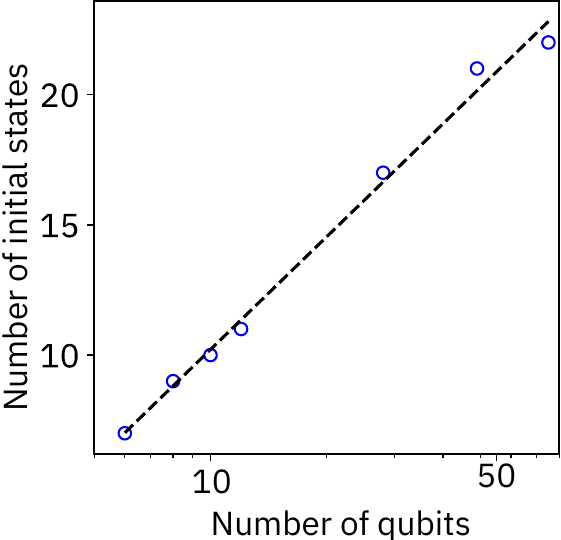}
\caption{ 
\textbf{Circuit sampling complexity for one-particle density (OPDM) reconstruction.} 
The number of initial states shown on a logarithmic scale needed to evaluate all matrix elements of the OPDM as a function of the number of qubits, for a single disorder realization. The relationship exhibits a close match to the logarithmic scaling curve, depicted as a dashed line.
}
\label{fig:OPDMmeasurements}
\end{figure}

\subsection{The algorithm}

Our selected model promotes an efficient use of IBM quantum hardware. In particular, each Floquet step requires only gate depth of two, measured in two-qubit gates, for one dimensional (1D) circuits. In two dimensional circuits (2D), the required gate depth is three. This is in contrast to a generic two-qubit gate which may require a depth up to six (or nine), as in the case of the first-order trotterized Heisenberg model.

The coupling between qubits is mediated by $U(\theta) = U_3 (\theta,0,\pi)$ gates (as defined by OpenQASM 3.0 convention). The coupling angle $\theta$ ranges from 0 to $\pi$, with $\theta\approx0.16\pi$ marking the finite-size MBL phase transition. Each $U(\theta)$ gate can be decomposed into two hardware native $SX$ gates and two virtual $RZ$ gates. 

After each layer of $CZ$ and $U(\theta$) gates, a phase gate $P(\phi_i)$ is implemented, representing a disorder $\phi_i$ for the qubit $i$. The disorder phases $\phi_i$ are randomly drawn from a uniform distribution between $-\pi$ and $\pi$. The total number of two-qubit $CX$ gates can be calculated as $(n-1)D$, where $n$ is the number of qubits and $D$ is the number of Floquet cycles. For example, for the 1D spin imbalance measurement on 104 qubits with $N_\mathrm{Floquet}=19$ Floquet cycles, this amounts to 1,957 gates in total.

The 2D circuits are similar to the 1D circuits, except that due to the connectivity of the 2D heavy hexagonal lattice, each Floquet cycle requires 3 layers of $CZ$ and $U(\theta)$ gates, where the layers correspond to the colors of the connections between qubits in Fig.~\ref{fig:2dsystem}a in the main text. The total number of gates is $N_\mathrm{edges}\times N_\mathrm{Floquet}$, e.g. the quantum circuit for the 2D spin imbalance measurement on 124 qubits with $N_\mathrm{edges}=139$ and $N_\mathrm{Floquet}=19$ Floquet cycles has 2,641 $CX$ gates.

\subsection{Initialization and measurement protocols}

Our experimental approach involves measurements aimed at reconstructing three distinct types of objects: spin imbalances, one-particle density matrices (OPDMs), and local integrals of motion (LIOMs). To obtain these, we employ the three-stage procedure on a quantum computer: state initialization, Floquet evolution, and the actual measurement process. In the following, we detail the methods used to measure each of these quantities, providing a comprehensive overview of our approach.
\\ 

\textbf{Spin imbalance}. 
Spin imbalance is defined to be
\be\label{eq:charge_imbalance}
\mathcal I \coloneqq \frac{n_1-n_0}{n_1+n_0}\;,
\ee
where $n_0$ and $n_1$ are the average occupation of the sites initialized in the zero and one states, respectively.
All spin imbalance measurements were performed on \texttt{ibm\_washington}. The initial layout for the 1D chain is  shown in Fig.~\ref{fig:logicalphysicalmapping}a. We initialize the system in a state where each qubit is set to the $|0\>$ state. Next, we set certain qubits (see Fig.~\ref{figs:cdw_state}) to the $|1\rangle$ state using an $X$ gate composed of two hardware-native and low-error (typically $\sim 3 \cdot 10^{-4}$ as measured by fidelity) $\sqrt{X}$ gates. 

For each value of the coupling angle ($\theta = 0.1\pi$ and $\theta = 0.3\pi$), we independently run 20 circuits with the number of Floquet cycles ranging from 0 to 19. For each cycle, at the end of the evolution, we measure the spin imbalance in Eq.~\eqref{eqs:imbalance} based on the expectation values of the qubits in the computational basis. Since qubits can be measured simultaneously, each spin imbalance measurement requires a total of 20 different circuits. Each circuit is run for 8,000 shots to collect the necessary statistics for the expectation. \\

\textbf{One-particle density matrix (OPDM).} All measurements are performed on \texttt{ibmq\_kolkata} (1D) and \texttt{ibm\_washington} (2D). We use quantum hardware to obtain the matrix elements of Eq.~\eqref{eq:opdm}, which can be expressed as
\be
\begin{split}
\rho_{ij} = &\frac{1}{4}\Bigl(\<X_iX_j\>_c-i\<X_iY_j\>_c+i\<Y_iX_j\>_c+\<Y_iY_j\>_c\Bigl),
\end{split}
\ee
where $\<AB\>_c = \<AB\>-\<A\>\<B\>$ is the connected correlation function, and $\<A\> := \<\psi_0|U_F^{d\dag}AU^d_F|\psi_0\>$ is the expectation of the operator with respect to the time-evolved initial ordered state $|\psi_0\>$.

For both, the 1D and 2D cases, we start by initializing the system in the geometry-dependent CDW state and then evolve for 10 Floquet cycles. For both cases, we measure qubits in the computational, or $Z$, basis, as well as various combinations of the $X$ and $Y$ bases. The number of combinations should be sufficient to determine all expectations of the operators $\<X_i\>$ and $\<Y_i\>$, as well as $\<X_iX_j\>$, $\<Y_iY_j\>$, $\<X_iY_j\>$, and $\<Y_iX_j\>$. It is important to note that measuring $O(\log(n))$ random combinations of $X$ and $Y$ bases is sufficient for this purpose \cite{elben2023randomized}. We further develop a simple optimization algorithm that requires a smaller number of multiqubit bases compared to random combinations. The sample complexity of the resulting measurement sequence is illustrated in Fig.~\ref{fig:OPDMmeasurements}.

For the one-dimensional circuits, we perform measurements in $\{7,9,10\}$ multi-qubit Pauli basis combinations for system sizes of $\{6,8,10\}$ qubits each, using a single initial state $|10...10\rangle$. The measurements are performed for coupling angles $\theta/\pi=\{0, 0.05, 0.10, 0.15, 0.20, 0.25\}$ and 100 disorder realizations, resulting in a total of 15,600 circuits. For the 2D OPDMs, we use $\{11, 17, 21, 22\}$ multi-qubit Pauli basis measurements for lattices with $\{1,3,6,10\}$ heavy hexagons (or $\{12, 28, 49, 75\}$ qubits) for the CDW initial state for the same set of coupling angles. The layout is shown in Fig.~\ref{fig:logicalphysicalmapping}b. The number of disorder realizations for system size containing $n$ qubits is calculated as $\lfloor 625/n\rfloor$, resulting in a total of 8,244 circuits.\\ 

\begin{figure*}[t!]
\centering
\includegraphics[width=1
\textwidth]{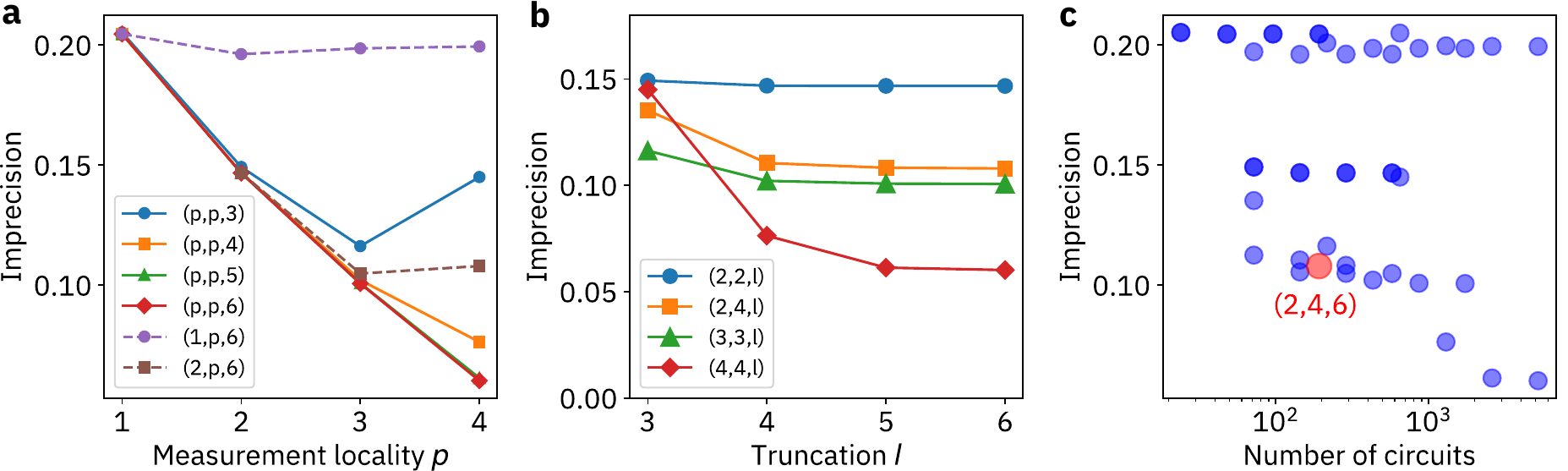}
\caption{\textbf{Cost-benefit analysis for LIOM reconstruction in the presence of noise.} This analysis uses a noiseless 10-qubit chain simulation for cycle depth $D=9$ with $L^{(0)}$ chosen as the Pauli $Z$ operator acting on the central qubit $5$ (qubits are numbered starting with 1). The results are averaged over 50 disorder realizations. \textbf{a.} The LIOM imprecision in Eq.~\eqref{eqs:operator_error} for various settings as a function of the maximum operator support. Solid lines represent the $(p,p,l)$ setting, dashed lines represent the $(p,k,l)$ setting. Increasing the locality leads to a drastic improvement in imprecision. \textbf{b.} The dependence of the error on the parameter $l$, which sets the number of initial states. \textbf{c.} The trade-off between the method error and the required quantum computational resources for $(p,k,l)$ schemes, where the dots represent different combinations of the settings $1\leq l\leq 4$, $1\leq p\leq 4$, and $1\leq l\leq6$. The red dot indicates the choice of the experimental setting ($(2,4,6)$).}
\label{fig:best_model}
\end{figure*}

\begin{figure*}[]
\centering
\includegraphics[width=1.0\textwidth]{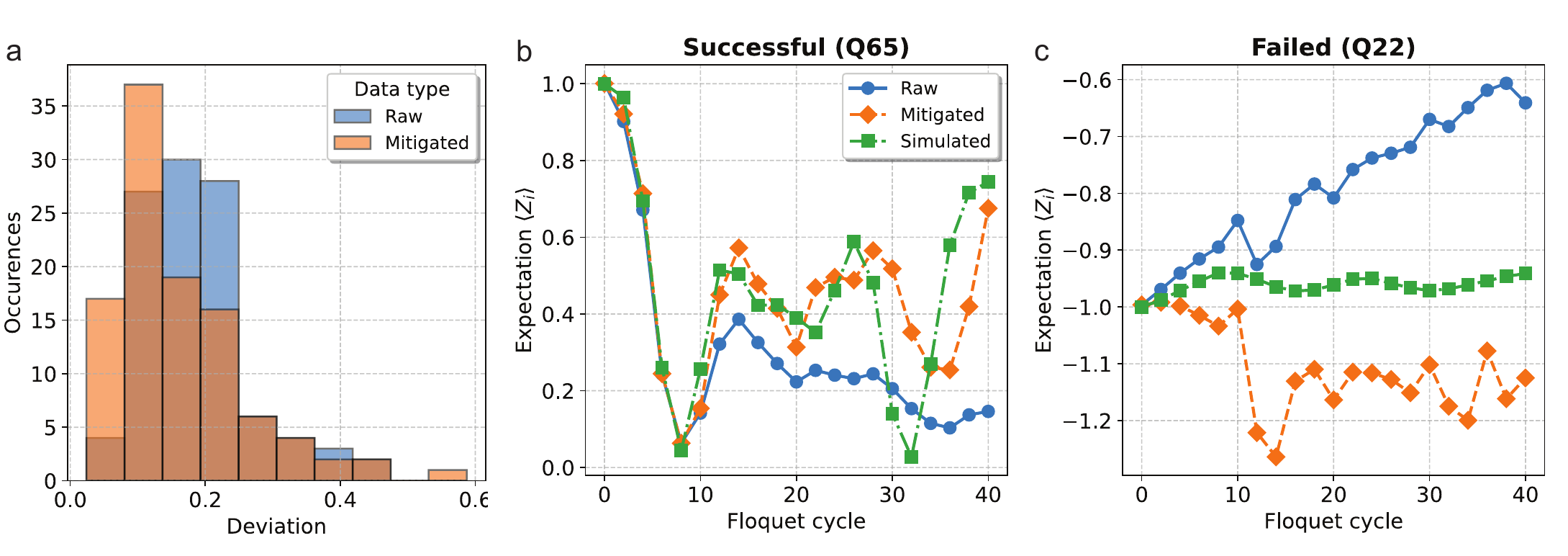}
\caption{\textbf{Benchmarking error mitigation.} This example illustrates the effect of error mitigation on the precision of quantum evolution for local expectations, $z_i(t) := \langle Z_i(t) \rangle$, using data from the \texttt{ibm\_kyiv} device, corresponding to Fig.~\ref{fig2:1dlocalization}b and c for $\theta = 0.1\pi$. \textbf{a.} The overall improvement in the signal $z_i(t)$ (both raw and mitigated) is quantified by the deviation of the time-dependent curves from the classical simulation of noiseless system $z_i^{\rm class}(t)$, i.e. $D_i = \sqrt{\mathbb{E}_t \left[z_i(t)^2 - z_i^{\rm class}(t)^2\right]}$, where the expectation is taken over the measurement cycles $t \in \{0, 2, 4, \dots, 40\}$. \textbf{b.} An example of successful error mitigation for qubit 65. The mitigated signal (diamonds, dashed lines) closely matches the classical simulation (squares, dash-dotted lines), compared to the original raw data (circles, solid lines). \textbf{c.} An example of failed error mitigation for qubit 22. The mitigated curve deviates significantly from the noiseless classical simulation and falls below the minimum physical value, $z_{\rm min} = -1$. Both successful and failed samples are included in the final results.}
\label{fig:benchmarkmitigation}
\end{figure*}

\begin{figure*}[t!]
\centering
\includegraphics[width=1
\textwidth]{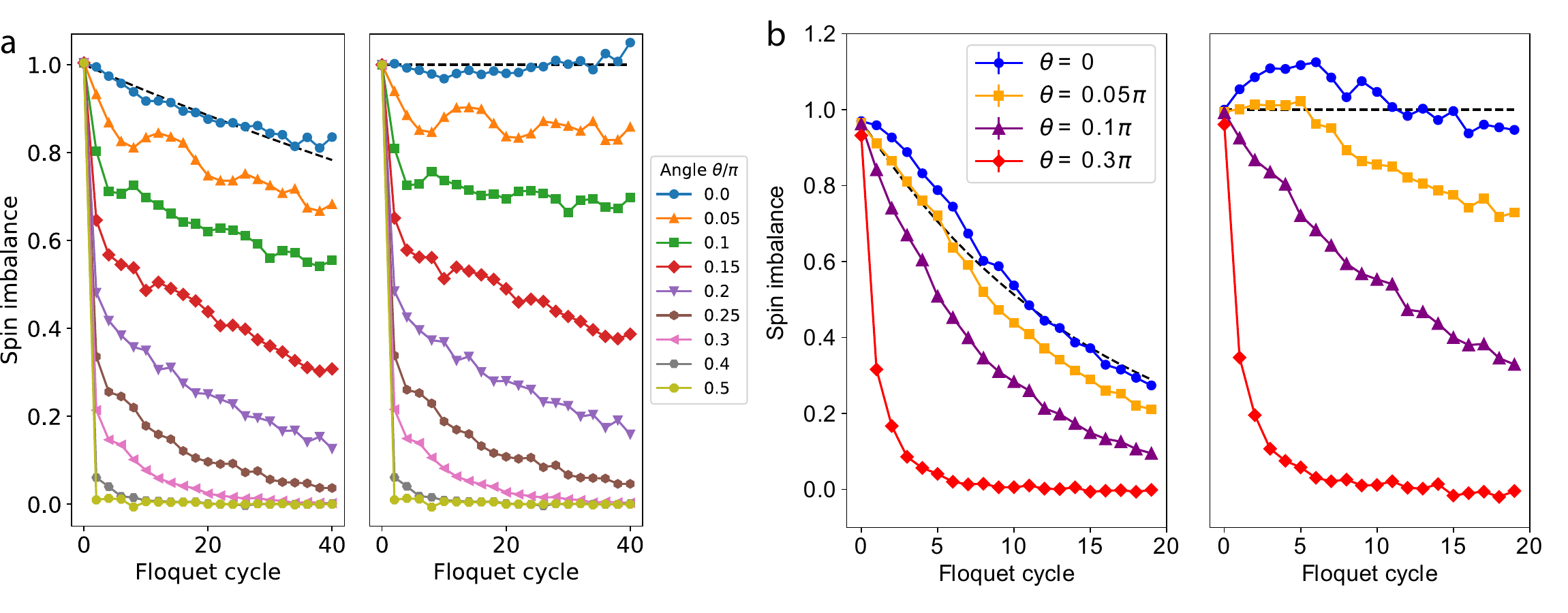}
\caption{ 
\textbf{Rescaling the charge imbalance.}  
\textbf{a.} The non-normalized (left) and the normalized (right) imbalance for a one-dimensional 104-qubit chain, also shown in Fig~\ref{fig2:1dlocalization}c. The non-normalized curves correspond to the error-mitigated results obtained from the device using Eq.~\eqref{eq:charge_imbalance}. The normalization is performed by the average slope of the $\theta=0$ curve (dashed curve). \textbf{b.} Similar curves for the 124-qubit hexagonal circuit.
}
\label{figs:full_cdw_dynamics}
\end{figure*}

\begin{figure}[t!]
\centering
\includegraphics[width=0.4
\textwidth]{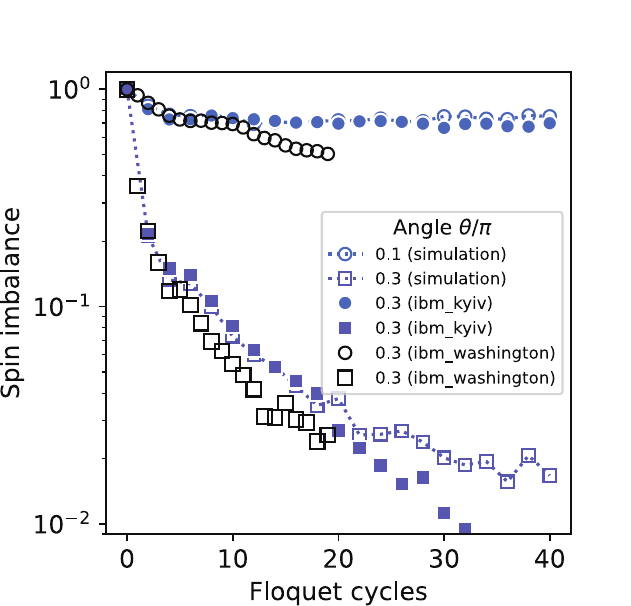}
\caption{ 
\textbf{Noise effect on spin imbalance.} 
The plot shows the relative comparison between classical numerical simulation (blue empty points, dashed lines), experiment on \texttt{ibm\_kyiv} (filled points) and \texttt{ibm\_washington}. The difference between the devices can be potentially explained by the noise effect.
}
\label{figs:compare_cdw_dynamics}
\end{figure}

\begin{figure*}[]
\centering
\includegraphics[width=1
\textwidth]{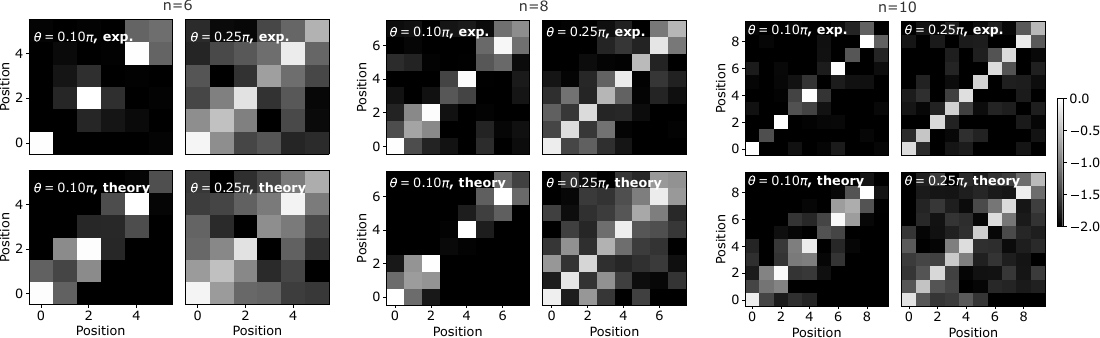}
\caption{\textbf{One-particle density matrices (OPDMs) in 1D lattices.} The color plots show the absolute values of the matrix elements $|\rho_{ij}|$ in Eq.~\eqref{eq:OPDM} for 1D chains with $n = 6$, 8, and 10 qubits after $10$ Floquet cycles.  The top panels show experimental data, while the bottom panels show theoretical predictions. The color scale is logarithmic, and the black color represents values smaller than $10^{-2}$.
}
\label{figs:1dOPDM_colorplots}
\end{figure*}

\begin{figure*}[t!]
\centering
\includegraphics[width=0.8
\textwidth]{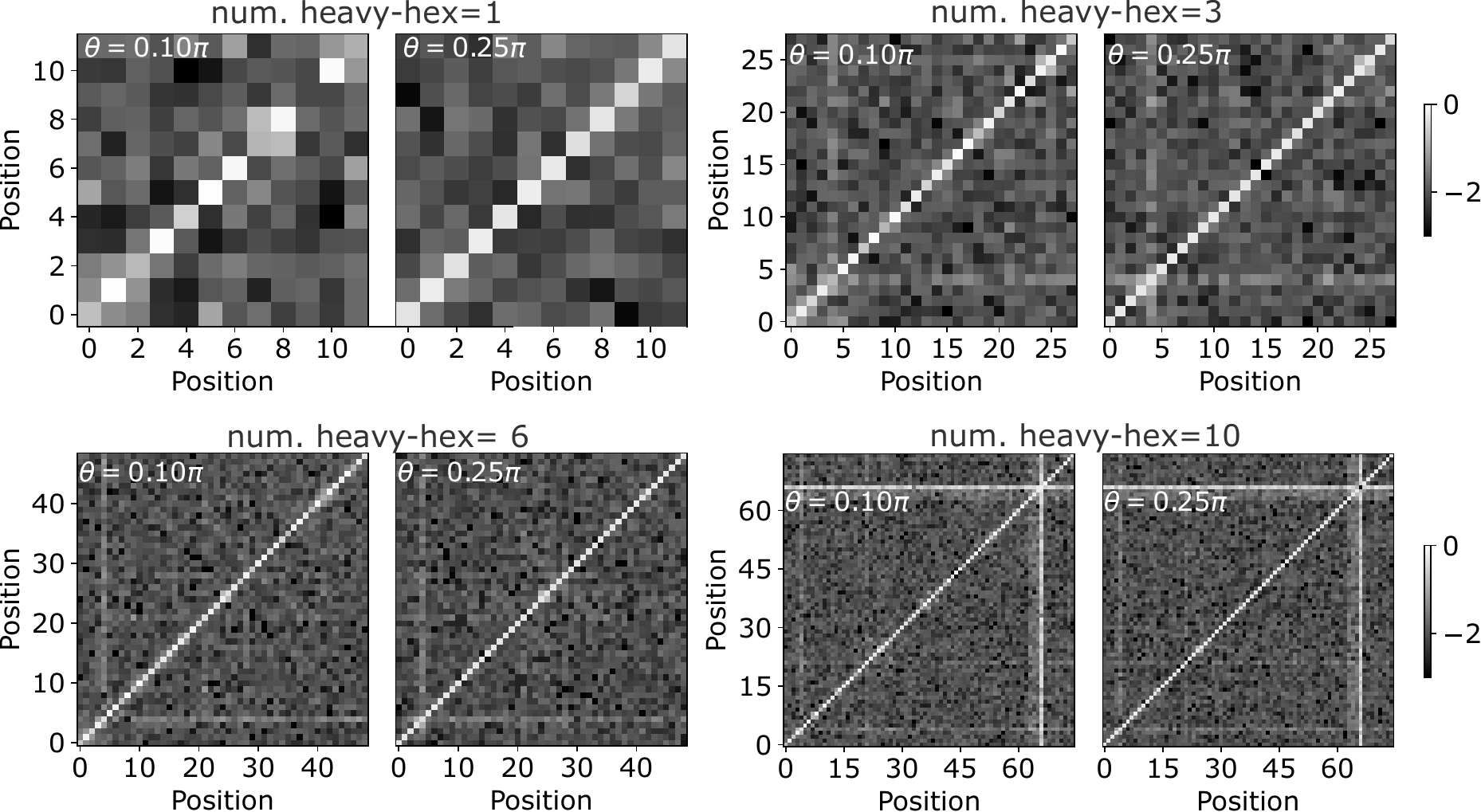}
\caption{\textbf{One-particle density matrices (OPDMs) in  2D lattices.}. We consider circuits containing 1, 3, 6, and 10 heavy hexagons for a state after $10$ Floquet cycles. The anomaly at position 66 for the OPDM for the 10 heavy hexagon circuit is due to a high readout error on this qubit. In the gap ratio and discontinuity shown in Figs.~\ref{fig:2dsystem}d and \ref{fig:2dOPDM_disc} respectively, we remove this entire row and column before calculating the spectrum. The color scale is logarithmic, and the black color represents values smaller than $10^{-3}$.
}
\label{figs:2dOPDM_colorplots}
\end{figure*}

\begin{figure}[]
\centering
\includegraphics[width=0.35\textwidth]{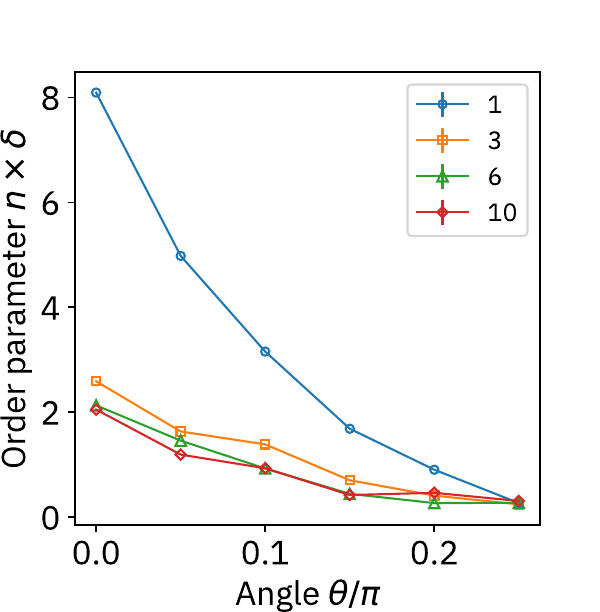}
\caption{\textbf{Discontinuity in the one-particle density matrix (OPDM) in 2D lattices.} The plot shows the rescaled discontinuity $n\delta$ multiplied by the number of qubits $n$ for two-dimensional circuits containing 1 (blue), 3 (orange), 6 (green), and 10 (red) hexagons. The parameter saturates to a system-independent value for more than one hexagon (see discussion in the text).}
\label{fig:2dOPDM_disc}
\end{figure}

\textbf{Local integrals of motion}. All measurements are performed on \texttt{ibm\_washington} (both 1D and 2D). To obtain the value of the correlation function in Eq.~\eqref{eqs:corr_funct_approx}, we measure the quantities
\be\label{eqs:f_parameter}
f_{\mu d \alpha} = \<\psi_\alpha|U_F^{d\dag}P_\mu U^d_F|\psi_\alpha\>,
\ee
where $|\psi_\alpha\>$ denotes the initial states, $P_\mu\in Q$ denotes a collection of local Pauli matrices that are part of a predefined set $Q$. Similar to previous experiments, the desired result can be achieved by running multiple runs of a quantum circuit. This process involves initializing the system in the state $|\psi_\alpha\>$, passing it through $d$ Floquet cycles, and then performing measurements of Pauli operators $P_\mu$. A major challenge in this procedure is to determine the optimal set of initial states $|\psi_\alpha\>$ to reliably estimate the value of $W_D(P_\mu, L^{(0)})$ in Eq.~\eqref{eqs:corr_funct_approx}.

The simplest and scalable method for selecting initial states is to assign them as random bit strings, using the Monte Carlo method. While this method is scalable (requiring $O(1)$ states to estimate the expectation of local operators), it is not an optimal use of resources for shallow circuits. In fact, to achieve an accuracy of $1\%$, we need about $10^4$ random initial states. Instead, we use a technique that allows us to achieve comparable accuracy, but with two orders of magnitude fewer initial states.

In one dimension, we choose a family of periodic states $|\psi_\alpha\> = |w_\alpha\>\otimes|w_\alpha\>\otimes \dots \otimes|w_\alpha\>$, where $|w_\alpha\>$ represents possible combinations of $l$ bit strings. The number of such states is only $2^l$. At the same time, for any evolution $U^d_F$ that transforms a local operator into an operator with support on $l$ qubits or less, expectation over this set of initial states is as effective as over all $2^{n}$ states.

For the two-dimensional heavy hexagonal lattice configuration, we find a similar suitable set of states. We first represent the Floquet cycle as a graph with vertices representing the qubits and the edges corresponding to next-neighbor connections. Next, we construct a randomized search algorithm that colors the graph using $l$ unique colors such that it minimizes the number of colors within $r$ steps of each qubit. The result is equivalent to the initial set of $2^l$ states, where each uniquely colored qubit represents an independently chosen value 0 or 1. As in the 1D case, the precision can be improved by choosing larger $l$.

Next, we develop a measurement protocol to get the values of the Pauli operators $P_\mu$. In one dimension, we choose to measure qubits at positions $i$, $i+p$, $i+2p$, and so on, in the same basis chosen from $X$, $Y$, or $Z$ for some integer $p\geq 1$. This pattern results in $3^p$ measurement combinations, which allows us to measure \textit{all local} $4^p - 1$ Pauli operators $P_\mu$ with support $p$. In addition, the protocol allows to measure \textit{some} of the Pauli operators with support $k>p$. We then measure all such Pauli operators and include them in the set $Q$.

To choose the Pauli measurement in two dimensions, we return to the graph representation of the Floquet cycle. We then use an algorithm that solves a graph coloring problem using $p$ colors so that each path containing $k$ qubits contains as many different colors as possible. Each such color corresponds to an independently chosen qubit measurement basis. This allows to measure all $k$-local Pauli operators if $p = k$, and fraction of them otherwise. For $p=2$, the best choice is the antiferromagnetic coloring shown in Fig.~\eqref{figs:cdw_state}b.

Combining the initialization protocol and the measurement protocol, we obtain a method we call the $(p,k,l)$ scheme. For this scheme, we use $2^l$ initial state combinations and $3^p$ Pauli measurement combinations to obtain all possible expectations of Pauli operators with support up to $k$. For any such scheme, the total number of unique quantum circuits required to run is $N_{\rm circ} = D\times 2^l\times 3^p$, where the factor of $D$ comes from the fact that we need to perform the measurements for each depth instance $d\in\{0,\dots,D-1\}$.

The analysis of different schemes in 1D, including their effect on the imprecision and the number of circuits required, is shown in Fig.~\ref{fig:best_model}.  For example, Fig.~\ref{fig:best_model}a and Fig.~\ref{fig:best_model}b show how the LIOM imprecision is reduced by increasing $p$ and $l$. While this reduction can be significant, the number of circuits grows exponentially with these parameters. To make the best choice, we perform the cost-benefit analysis shown in Fig.~\ref{fig:best_model}c. We found that the best scheme to achieve high precision with a reasonable number of circuits that exhibits sufficiently wide LIOMs is the $(2,4,6)$ scheme (marked by the red dot). In our experiment shown in Figs.~\ref{fig:1dLIOMs} and \ref{fig:2dsystem} we use this scheme including $N_{\rm circ} = 10\times64\times 9=5,760$ of non-mitigated circuits in total. This allows us to evaluate all values in Eq.~\eqref{eqs:f_parameter} for all two-qubit Pauli operators and some three- and four-qubit Pauli operators. We use the same scheme and parameters for the two-dimensional setting, although we are unable to estimate the precision classically.\\

\begin{figure*}[t!]
\centering
\includegraphics[width=1
\textwidth]{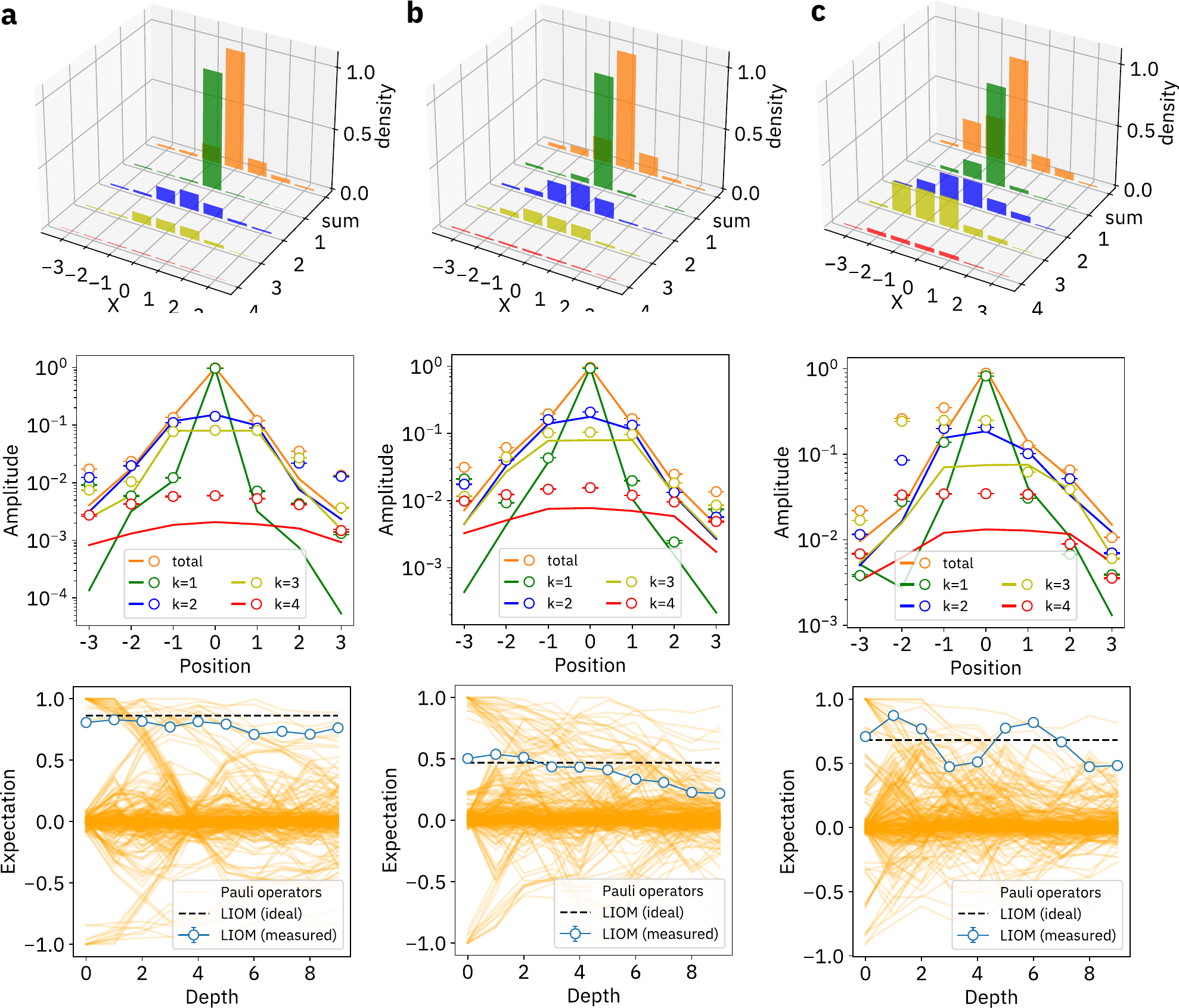}
\caption{
\textbf{Examples of local integrals of motion (LIOMs) with varying levels of error.}
The plot illustrates three LIOMs with different levels of imprecision: the lowest error (panel a, $\epsilon=0.065$), the median error (panel b, $\epsilon=0.135$), and the highest error (panel c, $\epsilon=0.357$). 
The top subpanels display the density of the LIOMs, while the middle subpanels present a logarithmic scale plot comparing the LIOM to the theoretical curve derived from the limited-size numerical simulation.
The bottom subpanels depict the time dependence of the involved Pauli operators (orange curves) and the expectation value of the LIOM itself. 
The dashed line represents the ideal behavior of the LIOM for infinite-depth evolution, as derived from the limited-size numerical simulation. 
All numerics use 12-qubit sub-regions to approximate the behavior of the LIOMs in the larger 104-qubit system.
\label{figs:LIOM_examples}
}
\end{figure*}

\begin{figure*}[t!]
\centering
\includegraphics[width=1\textwidth]{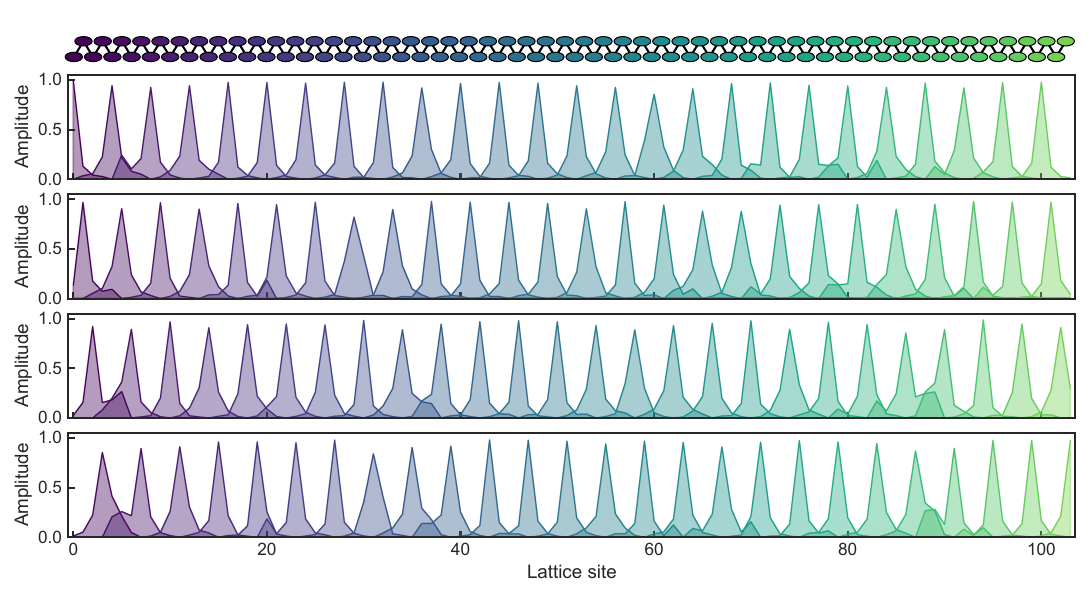}
\caption{
\textbf{Operator density representation for 104 LIOMs in a 1D spin chain.} This graph provides an experimental reconstruction of operator amplitudes for each of the 104 Local Integrals of Motion (LIOMs) in a 104-qubit chain, depicted schematically above the topmost plot. For improved readability and to prevent overlap, the LIOMs are distributed across four axes to enhance visual spacing. More specifically, the topmost axis displays the LIOMs at sites $i=0,4,8,\ldots$, the second from the top axis shows the LIOMs at sites $i=1,5,9,\ldots$, and so forth. The y-axes represent the operator spatial amplitudes~$W(x) = \sqrt{\sum_{\mu} a_{\mu}^2 w_{x\mu}}$, where $a_{\mu}$ is the reconstructed coefficients for the Pauli $\mu$ operator in a LIOM decomposition and $w_{x\mu}$ is the weight associated this Pauli operator at site~$x$.
\label{figs:LIOM_all_1D}
}
\end{figure*}

\begin{figure*}[t!]
\centering
\includegraphics[width=1
\textwidth]{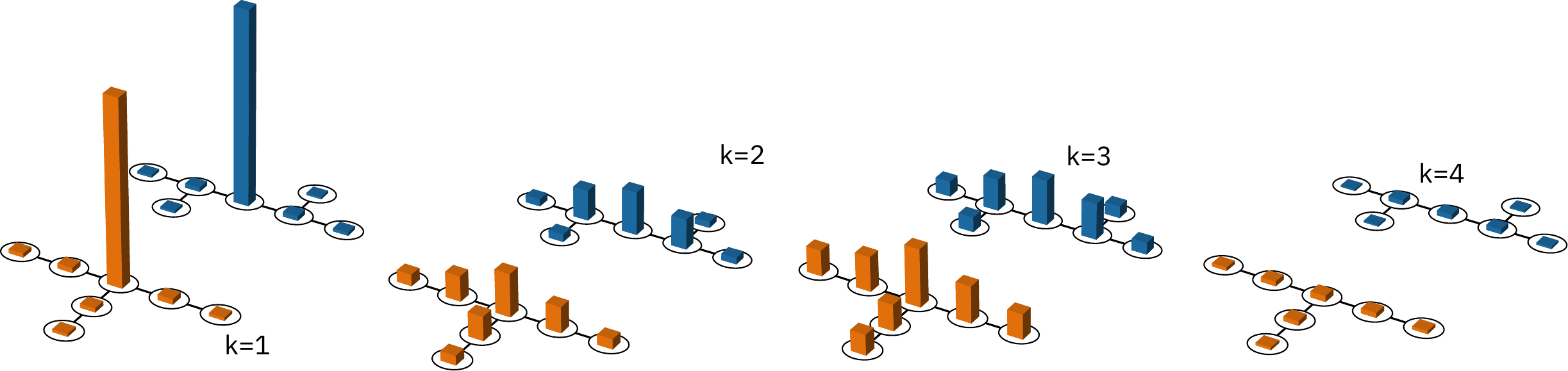}
\caption{
\textbf{Structure of local integrals of motion (LIOMs) in a 2D lattice.} 
The figure illustrates the decomposition of the ``average" LIOM into contributions from $k$-local Pauli operators, for $k$ in the range from 1 to 4. The blue bars represent the two-vertex configuration, while the orange bars represent the three-vertex configuration.
}
\label{figs:LIOM_2d_decomposition}
\end{figure*}

\subsection{Error mitigation}

Quantum error mitigation \cite{Ci2023E} describes strategies aimed at mitigating the detrimental effects of noise in quantum computations. Two popular and primary strategies are zero-noise extrapolation (ZNE) \cite{Li2017ZNE, Temme2017PECandZNE} and probabilistic error cancellation (PEC) \cite{Temme2017PECandZNE, Berg2022, Berg2022TREX}, which have recently been combined in \cite{Kim2023}.

Here, we developed a tailored, composite error-suppression and mitigation strategy that includes logical-to-physical qubit mapping~\cite{Nation2023Mapomatic}, M3 readout mitigation~\cite{Nation2021M3}, Pauli twirling~\cite{Wallman2017Twirling}, and ZNE~\cite{Li2017ZNE, Temme2017PECandZNE, majumdar2023best}. 
This strategy is straightforwardly scalable to large system sizes and accessible \textit{via} open-sourced packages offered through the Qiskit ecosystem. 

First, virtual qubits in the quantum circuit are mapped to physical qubits on the device. In the case of the 104-qubit 1D chain, the 124-qubit 2D lattice, and the 2D OPDM measurements, there are limited possibilities for this mapping. 
For the 1D OPDMs measured on 6, 8, or 10 qubits, the best qubits for a given circuit on a given device are selected using the Qiskit package \texttt{mapomatic}~\cite{Nation2023Mapomatic}. Based on the connectivity of the transpiled circuit, \texttt{mapomatic} finds possible mappings between virtual and physical qubits of a given device, ranks them based on a simple cost function that takes into account the characterized errors of each qubit and gate, and selects the best one. 

In order to mitigate coherent errors during runtime, we employ a combination of two strategies. The first strategy involves applying a technique known as ``twirling" to the noise present in the two-qubit gate layers of the circuits \cite{Wallman2017Twirling}. 
For a tutorial on twirling, see Ref.~\cite{Minev2022twirl}.
This approach entails the addition of random single-Pauli gates before and after the noisy two-qubit gates. These gates serve to ``mix up" the noise, helping to reduce its impact. Notably, the errors associated with the two-qubit gates are typically an order of magnitude larger than those of the single-qubit gates.
We randomly sample the twirled circuits $N_\mathrm{twirl}=5$ times and run each twirled circuit with $1/N_\mathrm{twirl}$ of the shots for the original circuit. The second strategy is dynamical decoupling (DD), which can suppress crosstalk errors \cite{Viola1999DD, Jurcevic2021}. Here, we insert evenly spaced XY4 pulses in idle periods for 1D single particle density matrix measurements, but do not use DD for spin imbalance and local integrals of motion measurements on larger systems as the brickwork quantum circuits have little idle time.

We also mitigate readout and incoherent errors by running additional circuits. Readout errors are mitigated on the noisy bitstring-level with the quasi-probability-based \texttt{M3} method \cite{Nation2021M3}. The ``balanced" version of the method requires $2N_\mathrm{qubit}$ readout calibration circuits to be run as opposed to the full matrix-inversion method that requires $2^{N_\mathrm{qubit}}$ calibration circuits~\cite{Bravyi2021ReadoutMatrixInversion}, which is infeasible for the system sizes considered here. We empirically find that the readout calibration need only be repeated every hour or so to maintain the performance of readout mitigation. Because many of the circuits---over many Floquet cycles, measurement bases, and initial states---in this study are run on the same qubits, running the readout calibration circuits incurs only a small overhead. Incoherent errors are mitigated with digital zero-noise extrapolation (ZNE)~\cite{Li2017ZNE, Temme2017PECandZNE}.
In dZNE, additional digital quantum gates that resolve to identity are inserted into the circuit to amplify the noise. Expectation values are computed at multiple noise factors, where a noise factor of 1 corresponds to the original circuit and, \textit{e.g.}, a noise factor of 3 corresponds to appending the inverse of the circuit followed by the circuit itself. Then, the zero-noise expectation value is estimated by extrapolating from the expectation values at these non-zero noise factors. In this study, we use noise factors of 1 and 3, and we use linear extrapolation, which generally underestimates the magnitude of observables but is more stable than exponential extrapolation.

We quantify the improvement in noisy expectation values in Fig.~\ref{fig:benchmarkmitigation} with and without quantum error mitigation on \texttt{ibmq\_kyiv} for a 104-qubit Floquet circuit with $\theta=0.1\pi$ and find that error mitigation allows the noise error to be reduced. More information on the performance of error mitigation in this regime for \texttt{ibm\_washington} can be found in Ref.~\cite{majumdar2023best}. For example, Fig.~9 there shows the simulation of a 104-qubit 1D chain with $\theta=0$ using a nearly identical error mitigation protocol -- the only difference being that noise factors of 1.0, 3.0, and 5.0 are used instead of just 1.0 and 3.0. Finally, error due to shot noise propagated through the error mitigation pipeline and data post-processing is quantified via bootstrapping over the bitstring counts.

Overall, this composite error mitigation strategy involves minimal additional classical compute to twirl the circuits and insert DD pulses, as well as minimal overhead of quantum resources---the total shot count of the entire composite strategy is just twice that of the unmitigated execution due to ZNE.

\begin{figure*}[t!]
\centering
\includegraphics[width=0.75\textwidth]{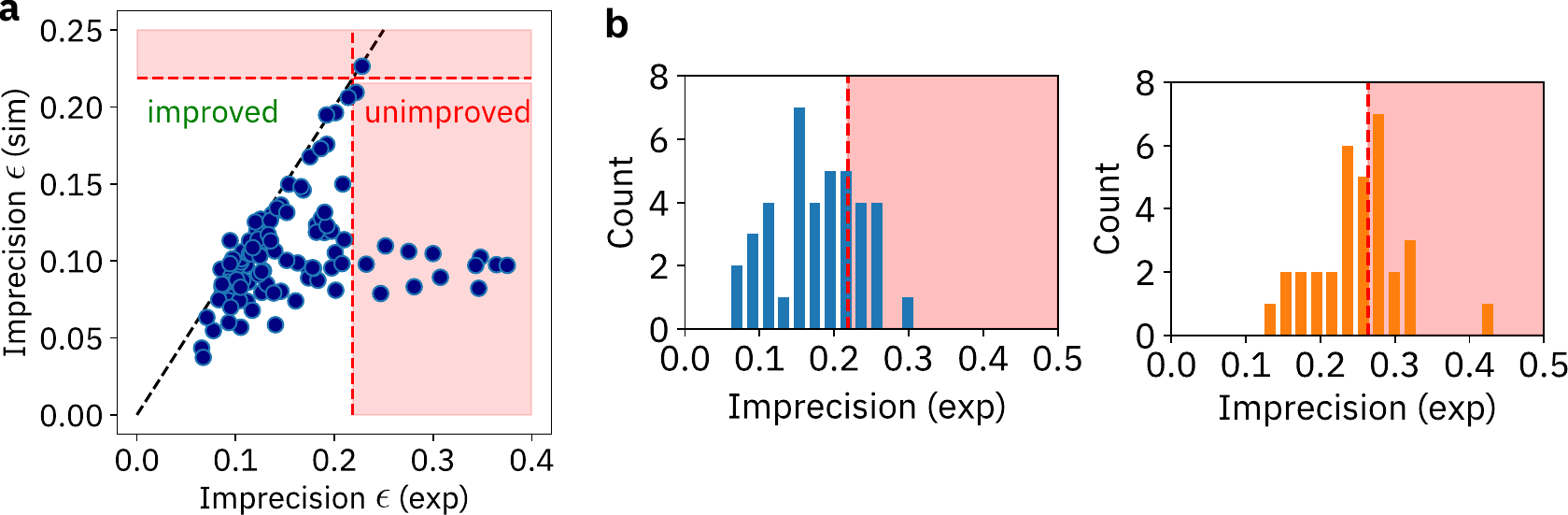}
\caption{
\textbf{Local integrals of motion (LIOMs) error due to noise.} \textbf{a.} Imprecision $\epsilon$ of the 1D LIOMs in the noiseless simulation (y-axis) and in the experiment (x-axis) for $\theta = 0.1\,\pi$. Red-shaded regions indicate error values exceeding the $\epsilon_0$ imprecision of the initial estimates~$L^{(0)}_k=Z_k$. \textbf{b.} Imprecision for 2D LIOMs. Two histograms show the distribution of $\epsilon$ imprecision for two-vertex (left) and three-vertex (right) prethermal LIOMs. The shaded region indicates imprecision exceeding that of the initial-guess operator.}
\label{fig:noiseerrorLIOMs2}
\end{figure*}

\subsection{Precision of the individual LIOMs}

In Figs.~\ref{fig:1dLIOMs}e and \ref{fig:2dsystem}f in the main text, we quantify the precision of the LIOMs using the localization length. The localization length is calculated for the $k$th LIOM using the definition
\be
\lambda_k = \sum_{\vec{x}} W(\vec{x})d(\vec{x},\vec{x}_0),
\ee
where $w_{k\vec{x}}$ is the $k$th LIOM amplitude at qubit $\vec{x}$ and $d(\vec{x},\vec{x}_0)$ is the Manhattan distance between qubits with coordinates $\vec x$ and $\vec x_0$. The relative error can be expressed as
\be
r_k = \frac{|\lambda^{\rm exp}_k-\lambda^{\rm sim}_k|}{\lambda^{\rm sim}_k},
\ee
where $\lambda^{\rm exp}_k$ and $\lambda^{\rm sim}_k$ are localization lengths derived from experiment and numerical simulation, respectively. There relative errors are shown in Fig.~\ref{fig:2dsystem}e by colors of data points.

We can also quantify the quality of the LIOMs by the normalized error in their commutator with the Floquet cycle unitary, 
\be\label{eqs:epsilon_err}
\epsilon = \|[L,U_F]\|_F/2\|L\|_F,
\ee
where $L$ is the LIOM operator and $\|\cdot\|_F$ is the Frobenius norm. Notably, numerical evaluation of this LIOM imprecision~$\epsilon_k$ is only polynomial in the number of Pauli terms and their locality. Moreover, it is  independent of the number of qubits for large systems. Note that, since we do not estimate coefficients for all Pauli operators, the derived quantity for a truncated LIOM serves as an upper bound on the its precision.
For our initial guess $L^{(0)}_k=Z_k$, the imprecision is~$\epsilon_0 = \sin\theta/\sqrt{2}$ ($\approx 0.2185$ for $\theta = 0.1\pi$) regardless of the qubit index~$k$ and disorder realization. Finding better precision compared to initial guess strongly indicates that the quantum algorithm works.

To distinguish system dynamics and noise effects, Fig.~\ref{fig:noiseerrorLIOMs2}a compares the imprecision of the LIOMs obtained from noise-free simulations (sim) and experiments (exp), yielding an average~$\epsilon = 0.16$. Absent experimental noise, all points would lie on the diagonal. 
Notably, several extracted LIOMs fail (dots in shaded area, representing 13\% of the total), attributed to outlier readout and gate calibration errors on a few qubits. Nonetheless, 87\% of the LIOMs surpass the initial prediction. Overall, noise contributes a 0.05 increase in net average imprecision compared to noiseless simulations, suggesting potential for improvement with enhanced devices and mitigation strategies \cite{Kim2023utility}. In comparison, for ergodic side of the dynamics ($\theta = 0.3\pi$) the initial guess yields $\epsilon_0 = 0.57$ and would not be substantially improved upon applying our method.

In Fig.~\ref{fig:noiseerrorLIOMs2}b, we report the experimental imprecision for each type of bulk LIOM for 2D setting.  Most LIOMs outperform the initial guess (bars on the left of the vertical dashed line) except for outliers associated with qubits experiencing poor performance. In 2D, larger connectivity and localization lengths lead to slower LIOM convergence and higher susceptibility to noise.

\section{Additional data}
\label{more_data}

In this section, we present additional experimental data to support the claims of the paper. Below, we simply list the additional results and provide the accompanying description.

\textit{Additional data on spin imbalance.} Figure~\ref{figs:full_cdw_dynamics} extends the plots presented in Figs.~\ref{fig2:1dlocalization}c and \ref{fig:2dsystem}c by showing non-normalized imbalance values and the dynamics of the circuits for both $\theta = 0$ and $\theta=0.05\pi$. The plots show a significant decay of the imbalance at $\theta=0$, even though the circuit must exhibit no dynamics if applied to a state in the computational basis. This decay can be attributed to the presence of noise. To mitigate the effect of noise, we perform a partial compensation by renormalizing the remaining curves using the average decay curve of $\theta = 0$, which is represented by the dashed black curve. The result of this normalization process can be observed in the right subpanels and is also shown in the main text. Also, Fig~\ref{figs:compare_cdw_dynamics} compares signal from two similar devices: \texttt{ibm\_kyiv} and \texttt{ibm\_washington}.

\textit{Additional data on OPDMs.}  Figure~\ref{figs:1dOPDM_colorplots} provides a visual representation of the OPDM for one-dimensional chains composed of 6, 8, and 10 qubits. This state has evolved over $D=10$ cycles for a specific disorder realization. In addition to the experimental data, the figure also shows theoretical predictions. In the ergodic regime ($\theta=0.3\pi$), the off-diagonal elements become more prominent, while the diagonal elements show a more uniform pattern compared to the MBL regime ($\theta=0.1\pi$). The average spectra resulting from the diagonalized OPDMs are shown in the main text in Fig.~\ref{fig2:1dlocalization}c.

Figure~\ref{figs:2dOPDM_colorplots} shows similar data for the OPDMs of 2D lattices containing 1, 3, 6, and 10 heavy hexagons. These correspond to 12, 28, 49, and 75 qubits, respectively. The corresponding circuit layouts for these configurations are shown in Fig.~\ref{fig:logicalphysicalmapping}b. As with the one-dimensional system, we consider a state that has evolved over $D = 10$ Floquet cycles for a singular disorder realization. The gap contrast of the spectra from the diagonalized OPDMs is shown in the main text in Fig.~\ref{fig:2dsystem}d.

Figure~\ref{fig:2dOPDM_disc} shows the rescaled discontinuity, similar to that shown in Fig.~\ref{fig2:1dlocalization}d for the one-dimensional circuit. However, unlike the one-dimensional system, this parameter does not exhibit a system-size scaling that depends on the coupling angle $\theta$ values for more than a single heavy hexagon. This phenomenon can be attributed to the strong influence of noise, which technically closes the hard gap in the OPDM spectrum for small $\theta$, resulting in a ``soft" gap (i.e. region of significantly reduced density without a real gap). As a result, the $n\delta$ parameter plateaus at a constant value that is independent of the system size. The thermalization process of the system as a function of system size can still be examined by observing the gap ratio in Fig.~\ref{fig:2dsystem}d, which is applicable to both hard and soft gaps.

\textit{Additional data on LIOMs}. Figure~\ref{figs:LIOM_examples} presents additional examples of individual approximate LIOMs for comparison. Here one can compare individual LIOMs with the `average' approximate LIOM shown in Figure~\ref{fig:1dLIOMs}b in the main text. For demonstration purposes, we have chosen three LIOMs with different levels of imprecision: the smallest, the median, and the highest. It is noteworthy that these operators have very similar shapes. We also plot the time dependence of each approximate LIOM as a function of cycle depth. As expected, the LIOM with the highest imprecision shows noticeable oscillations.

In Fig.~\ref{figs:LIOM_all_1D}, we present a depiction of all 104 Local Integrals of Motion (LIOMs), as experimentally reconstructed in the one-dimensional, 104-qubit experiment. Each LIOM is localized around a specific qubit site, exhibiting a distinct spatial extent and profile influenced by both the interactions and disorder realization. The amplitude of the weight assigned to each LIOM is plotted to showcase the varying spatial contributions across the entirety of weight Paulis. Due to the substantial overlap among the LIOMs, we have opted to distribute them across four axes for clearer visualization.

In Fig.~\ref{figs:LIOM_2d_decomposition}, we decompose the LIOM density shown in Fig.~\ref{fig:2dsystem}e into contributions from one-local to four-local Pauli operators. Similar to the one-dimensional scenario, the dominant contribution comes from the one-local Pauli operators, with a single Pauli~$Z$ acting on the central qubit having the largest contribution. Interestingly, the three-local contribution is nearly equal to the two-qubit contributions, suggesting quantum scrambling dynamics in the system.

\end{document}